\documentclass{amsart}

\usepackage{amsfonts, amsmath, amssymb, amsthm, bbm}
\usepackage[
backend=biber,
style=alphabetic,
citestyle=alphabetic
]{biblatex}
\usepackage{calc}
\usepackage[T1]{fontenc}
\usepackage[a4paper, total={6in, 8in}]{geometry}
\usepackage{graphicx}
\usepackage{hhline}
\usepackage[utf8]{inputenc}
\usepackage{multirow}
\usepackage{setspace}
\usepackage{verbatim}
\usepackage{float}

\theoremstyle{plain}
\newtheorem{theorem}{Theorem}[section]
\newtheorem{corollary}[theorem]{Corollary}
\newtheorem{example}[theorem]{Example}

\newtheorem{proposition}[theorem]{Proposition}
\newtheorem{remark}[theorem]{Remark}

\newcommand{\bft}[1]{\boldsymbol{f}_t}

\newcommand{\bsigmat}[1]{\boldsymbol{\sigma}_t}
\renewcommand{\d}{\,\mathrm{d}}
\newcommand{\eps}{\varepsilon}
\newcommand{\E}{\mathbb{E}}

\renewcommand{\i}{\ensuremath{\mathbbm{1}}}

\newcommand{\leli}{\ensuremath{\textrm{--}}}

\newcommand{\N}{\mathbb{N}}
\renewcommand{\P}{\mathbb{P}}

\newcommand{\R}{\mathbb{R}}
\renewcommand{\tilde}[1]{\widetilde{#1}}

\renewcommand{\quote}[1]{\guillemotleft#1\guillemotright}

\bibliography{bib}

\allowdisplaybreaks

\title{Hedging goals}
\date{\today}
\author{Thomas Krabichler}
\address{Eastern Switzerland University of Applied Sciences, Centre for Banking \& Finance}
\curraddr{}
\email{thomas.krabichler@ost.ch}
\thanks{}

\author{Marcus Wunsch}
\address{Zurich University of Applied Sciences, Institute of Asset \& Wealth Management}
\curraddr{}
\email{marcus.wunsch@zhaw.ch}
\thanks{}

\subjclass[2010]{65K99, 91G60, 91G20}

\begin{document}

\maketitle

\begin{abstract}
Goal-based investing is concerned with reaching a monetary investment goal by a given finite deadline, which differs from mean-variance optimization in modern portfolio theory. 
In this article, we expand the close connection between goal-based investing and option hedging that was originally discovered in~\cite{browne_reaching_1999} by allowing for varying degrees of investor risk aversion using lower partial moments of different orders.
Moreover, we show that maximizing the probability of reaching the goal (\textit{quantile hedging}, cf.~\cite{follmer_quantile_1999}) and minimizing the expected shortfall (\textit{efficient hedging}, cf.~\cite{follmer_efficient_2000}) yield, in fact, the same optimal investment policy. 
We furthermore present an innovative and model-free approach to goal-based investing using methods of reinforcement learning. 
To the best of our knowledge, we offer the first algorithmic approach to goal-based investing that can find optimal solutions in the presence of transaction costs. 
\end{abstract}

\section{Introduction}
\subsection{Motivation}
While modern portfolio theory~\cite{markowitz_portfolio_1952} posits that investors are risk-averse and thus should seek to maximize their portfolios' risk-adjusted returns, in reality investor often find themselves in need of capital to finance future investment goals: a car, an apartment, or their children's college education. 
The importance of investment goals on a societal level can be appreciated in view of the exacerbating retirement problem in many Western countries, cf.~\cite{giron_applying_2018}. 
Goal-based investment strategies are not primarily concerned with risk preferences relating to portfolio volatility; instead, 
they are subject to the risk of falling short of reaching a goal by its maturity. 
Even exceeding an investment goal is not necessarily desirable; in this case, a strategy with less volatility could have led to an outcome matching the investment goal.

There are at least two ways to translate this practical problem into a mathematical optimization problem. 
Either, one attempts to \textit{maximize} the probability of reaching an investment goal by a given maturity, or one tries to \textit{minimize} the expected shortfall (or a function thereof). 

\subsection{Literature review}
This first approach was investigated in a series of papers by Browne (cf.~\cite{browne_reaching_1999} and the references therein), who found the explicit portfolio allocation formula for the probability-maximizing strategy in the context of complete markets. 
In his articles, Browne used techniques from stochastic control theory as well as from PDEs. 
While highly appealing theoretically, the probability-maximizing paradigm suffers from the binary nature of its optimum: a goal missed by a hair's breadth is still a goal missed, and any such strategy will be discarded.
Rather, more and more leverage will be applied to attain the goal - even as the maturity draws closer -, resulting in either success or bankruptcy. 
This indifference for the \textit{size} of the shortfall constitutes a major drawback of probability-maximizing strategies for practical purposes. 

Leukert~\cite{leukert_absicherungsstrategien_1999} and F{\"o}llmer and Leukert (cf.~\cite{follmer_quantile_1999, follmer_efficient_2000, follmer_stochastic_2016}) treated the closely related problem of maximizing the probability of hedging contingent claims successfully when replication is attempted with less than the required initial capital (corresponding to the discounted value under the equivalent martingale measure). 
Their solution is based on a static optimization problem of Neyman-Pearson type. Another approach can be found in \cite{spivak_maximizing_1999}.

In practice, measuring and minimizing downward risk is arguably more significant than maximizing the probability of attaining a goal (in analogy with the dichotomy of Expected Shortfall versus Value-at-Risk, cf.~\cite{leukert_absicherungsstrategien_1999, follmer_stochastic_2016}).
Downward risk can be quantified by the shortfall, i.e., the positive part of the distance between the profit a strategy has earned at maturity and the goal. 
Several authors have addressed this problem in the context of replicating contingent claims, cf.~\cite{leukert_absicherungsstrategien_1999, follmer_quantile_1999, follmer_efficient_2000, pham_minimizing_2002, follmer_stochastic_2016}, including Cvitani{\'c} and Karatzas~\cite{cvitanic_dynamic_1999}. 
The latter authors employ tools from convex duality to show that a solution exists, and state explicit solutions for several special cases with a single risky asset. 
It is also interesting to note that quantile hedging~\cite{follmer_quantile_1999}, i.e., the probability-maximizing paradigm, can be interpreted as the most risk-seeking limit of efficient hedging, cf.~\cite{follmer_efficient_2000}. 
Nakano~\cite{nakano_efficient_2004} studied a similar problem, however considering coherent risk measures instead of lower partial moments (cf.~\cite{follmer_efficient_2000}).  

An intriguing and novel approach via optimal transport has recently been used to target prescribed terminal wealth distributions in~\cite{guo_portfolio_2021}. 


\cite{buehler_deep_2019} introduced a flexible framework for hedging contingent claims by applying deep learning methods. This approach transcends the classical Black-Scholes model's restrictions, e.g., the absence of transaction costs. 
Related reinforcement learning approaches can be found in~\cite{halperin_qlbs_2017, szehr_hedging_2021}. \cite{ruf_neural_2020} provides a comprehensive literature review regarding the application of neural networks for pricing and hedging purposes.

\section{Main contributions}
In our opinion, the potential that goal-based investing has for retirement saving and individual asset-liability-management cannot be overestimated. 

The theoretical foundations for the goal-based investment problem have been laid out in the - superficially unrelated - field of replicating contingent claims. 
Therefore, we regard adapting these results and making them accessible and palatable to practitioners as one of the main contributions of this paper. 
In particular, we show how risk preferences can be integrated into the original goal-based investment problem (cf., e.g., Proposition~\ref{prop:lowmom_ndim}), drawing on results on efficient hedging derived by F{\"o}llmer and Leukert~\cite{follmer_efficient_2000}. 

Another important contribution is the use of deep hedging techniques to deal with goal-based investing in a model-agnostic fashion. 
To the best of our knowledge, this is the first instance that transaction costs are incorporated into the optimization problems arising in goal-based investing. 

\section{Outline of this paper}


The remainder of this article is organized as follows. 
\\
After introducing the basic model in Section~\ref{sec:prel}, we state the optimal policy for risk-neutral and risk-taking goal-based investors in Section~\ref{sec:risk_neut}. 
The optimal policy for risk-averse goal-based investors, whose utility is determined by a lower partial moment of the shortfall relative to the goal, can be found in Section~\ref{sec:risk_aver}. 
We discuss shortcomings of the probability-maximizing paradigm in Section~\ref{sec:prac_impl}, where we also provide an illustrative example. 
To mitigate the risk inherent in the quantile and efficient hedging approaches, we propose a policy allowing for downward protection in Section~\ref{sec:down_prot}. 
\\
Finally, in Section~\ref{sec:deep_hedg}, we show that an artificial neural network can be trained to minimize the expected shortfall as well as lower partial moments, thereby approximating the optimal policies from Sections \ref{sec:risk_neut} and \ref{sec:risk_aver}. 
\\ 
The proofs of this paper can be found in Section~\ref{sec:proo} of the Appendix.

\begin{remark}
The explicit analytical results in Sections \ref{sec:risk_neut} - \ref{sec:down_prot} build upon the work in~\cite{browne_reaching_1999,follmer_quantile_1999, follmer_efficient_2000}. 
In particular, the validity of our analytical results is restricted to complete markets. \cite{follmer_quantile_1999, follmer_efficient_2000} also elaborate on the incomplete case based on duality results. 
Deep Hedging in Section~\ref{sec:deep_hedg} provides an appealing and highly flexible approach, as it is model-free and allows for the inclusion of transaction costs. 
\end{remark}

\section{Preliminaries}\label{sec:prel}
\subsection{The model}
We consider a \textit{complete} market with $n \in \mathbb{N}$ correlated risky assets generated by $n$ independent Brownian motions (cf.~\cite{browne_reaching_1999}), i.e.,
\begin{align}\label{eq:xit}
    \d X_t^{(i)} &= X_t^{(i)} \left[ \mu^{(i)}\d t + \sum_{j=1}^n \sigma^{(i,j)} \d W_t^{(j)} \right], \qquad i = 1, \dots, n, 
\end{align}
where the drift $\pmb \mu = \left(\mu^{(i)}\right)_{i=1}^n$ and the full rank volatility matrix $\pmb \sigma = \left(\sigma^{(i,j)}\right)_{i, j=1}^n$ are constant. 
$$\pmb{W}_t := \left(W_t^{(1)},\dots, W_t^{(n)}\right)^\top$$ shall denote a standard $n$-dimensional Brownian motion defined on the complete probability space $(\Omega, \mathcal{F}, \P)$
satisfying the usual conditions. 
\\
We assume that there is, in addition, a risk-less bank account compounding at the risk-free rate $r>0$, i.e.,
\begin{align} \label{eq:dBt}
    \d B_t=r\, B_t \d t,\qquad B_0=1.
\end{align}
The value of a zero-coupon bond at time $t$ that pays $1$ monetary unit at maturity $T > t \geq 0$ will be denoted as
\begin{align*}
    R_{t, T} &:= \exp\left\{ - r \ (T-t) \right\}. 
\end{align*}
We will only consider bonds without default risk. 
Monetary goals will be denoted by $H \in \R_+$ throughout. To ease notation, we shall write $H_{t, T} := R_{t, T} H$ for any $t \in [0, T]$. 
\\
We will make use of the diffusion matrix 
$\pmb \Sigma := \pmb \sigma \, {\pmb {\sigma}}^\top;$
the \textit{market price of risk} will be denoted by the vector $\pmb \vartheta$ defined as 
\begin{align}\label{eq:market_price_of_risk}
    \pmb \vartheta := {\pmb \sigma}^{-1} (\pmb \mu - r\pmb 1). 
\end{align}
We assume that all entries of $\pmb \vartheta$ are strictly positive. 
According to Girsanov's theorem, the vector process defined via the market price of risk as
\begin{align*}
    \pmb W^*_t &:= \pmb W_t + \pmb \vartheta \ t
\end{align*}
is an $n$-dimensional Brownian motion under the probability measure $\P^*$ given by its Radon-Nikodym derivative 
\begin{align}\label{eq:radon}
    \rho_* := \frac{\d \P^*}{\d \P\phantom{^*}} &= \exp\left\{ - {\pmb \vartheta}^\top \pmb W_T -\frac{1}{2} {\pmb \vartheta}^\top \pmb \vartheta\ T \right\}
    = \exp\left\{ - {\pmb \vartheta}^\top \pmb W^*_T +\frac{1}{2} {\pmb \vartheta}^\top \pmb \vartheta \ T \right\}, 
\end{align}
where $\P$ denotes the objective probability measure.  
The expectation under the risk-neutral measure $\P^*$ will be denoted as $\E^*$. 
\\
The optimal growth portfolio (cf.~\cite{platen_role_2005}), which is also known as the \textit{market portfolio}, maximizes the growth rate of wealth; cf.~\cite[][Subsection 4.2]{browne_reaching_1999}. 
Its weights $\pmb{\pi}_*$ and its variance $\sigma_*$ are determined as 
\[
    \pmb{\pi}_* := \left(\pmb{\sigma}^{-1}\right)^\top\, \pmb{\vartheta}_t,\qquad
    {\sigma_*}^2 := {\pmb{\pi}_*}^\top \, \pmb{\Sigma}\, \pmb{\pi}_* = \pmb{\vartheta}^\top \pmb{\vartheta} = \sum_{i=1}^n \left(\frac{\mu^{(i)}-r}{\sigma^{(i,i)}}\right)^2.
\]
The optimal growth portfolio evolves as (cf.~\cite[][Subsection 4.2]{browne_reaching_1999})
\begin{align*}
    \Pi_t &= \Pi_0 \exp\left\{ \left(r - \frac{1}{2} {\sigma_*}^2 \right) t + {\pmb\vartheta}^\top \ \pmb{W}_t^* \right\}.
\end{align*}

\begin{remark}
For ease of notation, we use constant coefficients throughout this article. It is, however, straightforward to generalize all our results to deterministic time-dependent coefficients. 
\end{remark}

\subsection{Goal-based investing \& hedging}
The goal-based investment problem can be expressed in terms of replicating a contingent claim with a constant payoff at maturity $T>0$ given by $H>0$, cf.~\cite{browne_reaching_1999}.
It is thus equivalent to finding an admissible strategy $(V_0, \,\pmb \xi)$, given for $t\in[0,T]$ by 
\begin{align}\label{eq:Vt}
    V_t = V_0 + \int_0^t {\pmb \xi_s}^\top \d \pmb W_s,
\end{align}
where $\pmb \xi$ is a predictable process with respect to the Brownian motion $\pmb W$ such that 
\begin{align}\label{eq:expsf}
    \mathbb{E}[ \ell ( (H - V_T)_+ ) ] 
\end{align}
becomes \textit{minimal}. Here, the expectation $\E$ is taken under the objective probability measure $\P$, and $\ell$ denotes a loss function that expresses the risk appetite of the investor. 
We will consider loss functions of the type \begin{align}\label{eq:elloss}
    \ell_p(x) = x^p, \quad p \in \R_{\ge 0}.
\end{align} 
For these loss functions, the expression \eqref{eq:expsf} is referred to as the \emph{lower partial moment of order} $p$.
Note that, as $p\rightarrow 0+$, the integrand in \eqref{eq:expsf} tends to the indicator function $\mathbbm{1}_{(0, H)}(V_T)$. 
This situation is tantamount to quantile hedging as discussed in~\cite{follmer_quantile_1999}. 
Conversely, risk aversion increases as $p\rightarrow\infty$.
\\
Let us assume that the investor initially posts the amount $V_0 = z > 0$. 
If $z$ is such that $z \ge H_{0, T}$, then the zero-coupon bond can be perfectly replicated at no risk, and the expected loss \eqref{eq:expsf} vanishes. 
On the other hand, if $z < H_{0, T}$, then the investor faces the risk of falling short of her desired goal. 

\section{Risk Neutrality and Risk Taking}\label{sec:risk_neut}
The policy minimizing the expected shortfall for hedging a zero-coupon bond paying out $H\equiv 1$ at maturity was derived in~\cite{xu_minimizing_2004}. 
In what follows, we extend her approach to incorporate a constant risk-free rate $r>0$ and an arbitrary constant payoff $H\in\R_+$ such that $z < H_{0, T}$.
Moreover, we show that the result of \cite[][Section 2.2.1]{xu_minimizing_2004} is, in fact, equivalent to the one of \cite{browne_reaching_1999} for $H \equiv 1$. 
In particular, the hedging strategy in the case of a single risky asset is indeed independent of its drift, which is not immediately obvious from the formulae stated in \cite{xu_minimizing_2004}. 
\begin{remark}
In the following discussion, we treat the entire spectrum of risk appetites ranging from risk neutrality ($p=1$, also referred to as \emph{efficient hedging}) to extreme risk taking ($p=0$, also referred to as \emph{quantile hedging}).
The theoretical foundations can be found in Subsection~5.4 of \cite{follmer_efficient_2000}. 
The discussion in Section~\ref{sec:risk_aver} will address higher degrees of risk aversion by considering lower partial moments of order $p > 1$. 
\end{remark}
\begin{proposition}[Efficient hedging using several risky assets]\label{prop:effhed_ndim}
Consider an investment with an initial capital endowment of $z$ monetary units, whose objective is to minimize 
\[
    \E\left[{(H-V_T)_+}^p\right], \qquad H \in \R_+, \qquad p \in [0, 1].
\]
Then the optimal policy for this objective is equivalent to the replication of a European digital call option on the optimal growth portfolio $\Pi_t$ with payoff $H$ and strike price $K^*$, where 
\begin{align}\label{eq:strike_ndim}
    K^* &= \Pi_0\; \exp\left\{ \left( r-\frac{1}{2}{\sigma_*}^2\right)\,T - \sigma_* \ \sqrt{T}\ \Phi^{-1}\left(\frac{z}{H_{0, T}}\right) \right\}, 
\end{align}
$\Phi$ denotes the cumulative distribution function of the standard normal distribution, and $\Phi^{-1}$ the corresponding quantile function. \\
In particular, the investor's wealth process can be expressed by means of
\begin{align}\label{eq:digicall_ndim}
    V_t &= H_{t, T}\, \Phi\left(\frac{\log \frac{\Pi_t}{K^*} + \left( r-\frac{1}{2}{\sigma_*}^2\right)\,(T-t)}{\sigma_* \sqrt{T-t}}\right).
\end{align}
\end{proposition}
\begin{remark}
Note that, if $z = H_{0, T}$, then the strike $K^*$ given in  \eqref{eq:strike_ndim} will vanish. As a consequence, the value of the standard normal distribution function $\Phi$ in \eqref{eq:digicall_ndim} will be 1, so that the claim reduces to a risk-less bond, $V_t = H_{t, T}$. 
If $z$ is even larger than the discounted goal, compounding will result in super-replication. 
\end{remark}
\begin{corollary}[Efficient hedging using a single risky asset]\label{cor:effhed_1dim}
In the case of a single risky asset, the contingent claim \eqref{eq:digicall_ndim} can be simplified to 
\begin{align*}
    V_t &= H_{t, T}\, \Phi\left(\frac{\log \frac{X_t}{K^*} + \left( r-\frac{1}{2}{\sigma}^2\right)\,(T-t)}{\sigma \sqrt{T-t}}\right),   
\end{align*}
where 
\begin{align*}
    K^* &= x_0\, \exp\left\{ \left(r-\frac{1}{2}\sigma^2\right) T - \Phi^{-1}\left( \frac{z}{H_{0, T}} \right) \right\}. 
\end{align*}
The corresponding delta-hedging strategy is obtained by differentiation: 
\[
    \xi_1(t, X_t) 
    = 
    \frac{\partial}{\partial x}\bigg \rvert_{x = X_t}V_t
    = \frac{H_{t, T}}{X_t\ \sigma\sqrt{T-t}}\ \phi\left( \frac{\log \frac{X_t}{K^*} + \left(r-\frac{\sigma^2}{2}\right)(T-t)}{\sigma \sqrt{T-t}} \right), 
\]
where $\phi$ denotes the probability density function of the standard normal distribution. 
\end{corollary}


\begin{corollary}
In the case of a constant claim $H \in \R_+$, the optimal policies for quantile hedging~\cite{follmer_quantile_1999} and efficient hedging~\cite{follmer_efficient_2000} coincide. 
\\
In particular, \cite[][Corollary 2.8]{xu_minimizing_2004} concerning the efficient hedging of a bond with payoff $H\equiv 1$ yields the same optimal policy as \cite[][Proposition 4.1]{browne_reaching_1999} with goal $b\equiv 1$ and vanishing risk-free rate. 
\end{corollary}

\section{Practical Considerations when Maximizing Probabilities} \label{sec:prac_impl}

\subsection{The Classical Notion of Goal-Based Investing}
Let us assume that the investment universe consist of a single risky company share $X=(X_t)_{t\in[0, T]}$ and a risk-free bank account $B=(B_t)_{t\in[0, T]}$, cf.~\eqref{eq:xit}, \eqref{eq:dBt}. A digital (or binary) European call option on the underlying $X$ with strike $K>0$ is a financial derivative with payoff $\i_{\{X_T\geq K\}}$ at maturity $T$. Its Black-Scholes price is given by
\begin{align}\label{eq:bs_price_digital}
    C(t;X_t,K) &=R_{t,T}\ \Phi(d_-(t;X_t,K))
    \\
    d_-(t;x,K)&:=\frac{\log{\frac{x}{K}}+\left(r-\frac{\sigma^2}{2}\right)(T-t)}{\sigma\sqrt{T-t}}. \nonumber
\end{align}
According to Corollary~\ref{cor:effhed_1dim} (cf.~\cite[][Section 4]{browne_reaching_1999}), continuous re-balancing with
\[
\xi(t; X_t, K)=\frac{R_{t, T}}{X_t\, \sigma\, \sqrt{T-t}}\ \phi\left(d_-(t;X_t,K)\right)
\]
replicates this digital payoff starting from $V_0=C(0;X_0,K)$
monetary units. 
By inspection, the initial price $V_0=V_0(K)$ is monotonously decreasing with
\[\lim_{K\to0+}V_0(K)=R_{0, T},\qquad\lim_{K\to\infty}V_0(K)=0.\]

Let us assume that a financial investor owns $c_0>0$ monetary units at time $t=0$ and, by means of an admissible 
strategy in the investment universe, aims at owning 
$c_T > c_0$ monetary units at time $T$. 
For simplicity, let us exclude intermediate income and consumption. 
In order to ensure that the mathematical problem is well-posed, one needs to establish in what sense a certain strategy becomes optimal. 
In~\cite[][Theorem 3.1]{browne_reaching_1999}, Browne proved the intriguing fact that replicating $c_T$ digital call options with strike
\[
K^*=X_0 \exp\left\{\left(r-\frac{1}{2}\sigma^2\right)T-\sigma\,\sqrt{T}\,\Phi^{-1}\left(\frac{c_0}{R_{0, T}\,c_T}\right)\right\}
\]
maximizes the objective probability of reaching the goal.
This result has an insightful economic interpretation; $K^*$ coincides with the break-even point with respect to the strike where a single digital call option costs $\frac{c_0}{c_T}$ at time $0$. Notably, but also well known, the magnitude of the hardly ascertainable drift $\mu$ does not affect $K^*$. In fact, the above expression of $K^*$ is only well-defined provided that the argument of $\Phi^{-1}$ is within $(0,1]$. In our setting, this prerequisite is only violated in the degenerate case $c_0 \geq R_{0, T}\,c_T$, i.e., the goal can be super-replicated in terms of the bank account at no risk anyway. 
The maximized real-world probability of reaching the goal is
\[
\P\left[X_T\geq K^*\right]=\Phi\left(\vartheta \sqrt{T}+\Phi^{-1}\left(\frac{c_0}{R_{0, T}\,c_T}\right)\right).
\]

For real-world applications, the financial investor has two alternatives; either she buys it over-the-counter or she replicates the digital payoff herself. In the former case, she runs the risk of not getting the promised payoff due to the bankruptcy of the issuer. In the latter case, without further stop-loss measures in place, discrete re-balancing schedules imply the risk of arbitrarily large losses way beyond $c_0$ due to the discontinuity of the payoff and, hence, the unbounded delta of the digital option. Notably, the strategy also requires an unlimited credit line at the bank which is collateralized only to an insufficient extent by the company share. Transaction costs exacerbate the situation. By approximating the digital payoff by a classical bull call spread and by diversifying the involved derivatives across several bona fide counterparties, the financial investor manages to deal with the mentioned impediments all the same. 
From a computational perspective, we lose analytical tractability with increasing degrees of complexity, e.g., additional constraints, more realistic price dynamics, transaction costs, etc. 
Despite all, and much more crucially, the \quote{all-or-nothing} feature of the proposed optimal strategy is not feasible in many real-world applications such as traditional pension funds. For obvious reasons, retirement savings are not supposed to be a Bernoulli experiment. 
Therefore, we will consider further ways to control downward risk in Section~\ref{sec:down_prot}.

\begin{example}
Let us consider a simple one-step financial market that hosts two financial assets over the time horizon $t\in\{0,1\}$. For some $0<\eps\ll 1$, a risk-free bank account carries a deterministic log-return of $r-\eps$ for some $r\in\R$. The other investment alternative is a start-up company whose success is dichotomous; the log-return $\tilde{r}$ of the company share satisfies $\P[\tilde{r}=r-1]=p$ and $\P[\tilde{r}=r+1]=1-p$ for some $p\in(0,1)$. Let $\xi\in[0,1]$ denote the portion of the initial wealth that is kept in the risky asset. The log-return of any strategy $\xi$ is then given by $R(\xi)=\log\big(\xi e^{\tilde{r}}+(1-\xi)e^{r-\eps}\big)$. From a practitioner's perspective, if the investor's ultimate goal was to reach a continuously compounded yield of $r$, then it would not be advisable to invest in the risky asset at all. However, a strict application of maximising the probability of reaching the goal would involve shortfall risk. Indeed, it holds $\P\left[R(0)\geq r\right]=0$, whereas $\P\left[R(\xi)\geq r\right]$ is maximal for any
\[
\xi\geq\frac{e^\eps-1}{e^{1+\eps}-1}.
\]
\end{example}

This example shows that the probability-maximizing paradigm might be too rigid in the context of goal-based investing as it does not take into consideration the investor's risk appetite. 
In the next section, we will discuss optimal policies for risk-averse investors. 

\section{Risk Aversion} \label{sec:risk_aver}
We consider the case of $p>1$, so that $(\ell_p)_{p>1}$ (cf.~\eqref{eq:elloss}) denotes a series of convex loss functions corresponding to increasing levels of risk aversion as $p$ grows. 
According to~\cite[][Lemma 11]{leukert_absicherungsstrategien_1999}, the optimal strategy to minimize \eqref{eq:expsf} consists in hedging the modified claim
\begin{align}\label{eq:modifH}
    \varphi_p H &= H - \min\left(a_p\, \rho_*^{\frac{1}{p-1}}, H\right), 
\end{align}
where the constant $a_p$ is implicitly determined by the capital requirement $\mathbb{E}^*[\varphi_p\, H] = z$.

\begin{proposition}[Risk aversion with several risky assets]\label{prop:lowmom_ndim}
Consider an investor endowed with $z$ monetary units at time $t = 0$. We assume that her objective is to minimize the lower partial moment 
\[
\E\left[{(H-V_T)_+}^p\right], 
\]
for $p>1$, cf.~\eqref{eq:expsf}.  
Then, the optimal strategy is equivalent to replicating the contingent claim on the optimal growth portfolio $\Pi_t$ with value process
\begin{align*} 
\begin{split}
    V_t = V(t, \Pi_t) = H_{t, T}\, \bigg\{&\Phi(d_-(t; \Pi_t, L)) - \bigg(\frac{L}{\Pi_t}\bigg)^{p'} \exp\bigg\{  p' (p'+1) \bigg(\frac{1}{2}{\sigma_*}^2-r\bigg) (T-t) \bigg\} \times
    \\
    & \Phi\bigg( d_-(t; \Pi_t, L)-p' \,\sigma_*\, \sqrt{T-t} \bigg)\bigg\}. 
\end{split}
\end{align*}
Here, $p'=1/(p-1)$, and the threshold $L$ is implicitly determined by the capital requirement $V(0, \Pi_0) = V_0 = z$. 
\end{proposition}

\begin{corollary}[Risk aversion with a single risky asset]\label{cor:lowmom_1dim}
If there is only one risky asset $X=(X_t)_{t\in[0, T]}$ available to the investor, then the optimal strategy to minimize the lower partial moment \eqref{eq:expsf} with exponent $p>1$ will be equivalent to replicating the value process $V_t=V(t,X_t)$ equal to
\begin{align}\label{eq:digiCall_val_proc}
H_{t, T}\,\bigg\{\Phi(d_-(t; X_t, L)) - \frac{L^{\alpha_p}}{X_t^{\alpha_p}}
    \exp\bigg\{  \alpha_p (\alpha_p+1) \bigg(\frac{1}{2}\sigma^2-r\bigg) (T-t) \bigg\}\, \Phi\bigg( d_-(t; x, L)-\alpha_p\, \sigma\, \sqrt{T-t} \bigg)\bigg\},
\end{align}
where $\alpha_p := \alpha/(p-1)$ and $\alpha := (\mu-r)/\sigma^2$. The hedging strategy $\xi_p$ is given by 
\begin{align*}
    \xi_p(t, X_t) 
    = H_{t, T} \Bigg(&\frac{\phi(d_-(t;X_t,L))}{X_t \sigma \sqrt{T-t}}\\
    &\quad - \frac{L^{\alpha_p}}{X_t^{\alpha_p}}\exp\left\{  \alpha_p \left(\alpha_p+1\right) \bigg(\frac{1}{2}\sigma^2-r\bigg) (T-t) \right\} \frac{\phi(d_-(t;X_t,L)-\alpha_p\,\sigma\,\sqrt{T-t})}{X_t\, \sigma\, \sqrt{T-t}}
    \\ 
    &\quad+ \frac{\alpha_p L^{\alpha_p}}{X_t^{\alpha_p+1}}\exp\left\{ \alpha_p \left(\alpha_p+1\right) \left(\frac{1}{2}\sigma^2-r\right) (T-t) \right\} \Phi(d_-(t; X_t, L) - \alpha_p\,\sigma\,\sqrt{T-t}) \Bigg).
\end{align*}
\end{corollary}
\begin{remark}
The first term in the expression for the modified claim $\varphi_p H$ in \eqref{eq:digiCall_val_proc} constitutes a digital European call option with strike $L$ and terminal payoff $H\mathbbm{1}_{\{ X_T \geq L \}}$.
\end{remark}
\begin{remark}
From a practical viewpoint, plausible values for $\alpha$ would be around $1$, assuming $\mu = 5\%$, $r=1\%$, and $\sigma = 20\%$. 
The exponent $\alpha_p$ would then be positive and decrease from $1$ to $0$ as $p \rightarrow \infty$ $(p>1)$.
\end{remark}
\begin{remark}
If the term corresponding to a digital European call option in Equation~\eqref{eq:digiCall_val_proc} matures in-the-money (i.e., $X_T > L)$, then the second term in this equation equals $(L/X_T)^{\alpha_p}$, which is less than $1$ and decreasing in $X_T$ if $\alpha_p > 0$. Conversely, if the digital call expires at-the-money, the second term in Equation~\eqref{eq:digiCall_val_proc} will be $1$, so that the entire claim matures worthless. 
The same holds true if the digital call expires out-of-the-money. 
\end{remark}
\begin{remark}
What happens in the case of extreme risk aversion, i.e., as $p\rightarrow\infty$? 
By Equations~\eqref{eq:modifH} and~\eqref{eq:fraclx}, 
\[
\lim_{p\to\infty}a_p=H-R_{T,0}z,\qquad\left(\frac{L}{X_T}\right)^{\alpha_p}=a_p\frac{k^{\frac{1}{p-1}}}{{X_T}^{\frac{\alpha}{p-1}}}.
\]
Hence,
\[
\lim_{p\to\infty}\varphi_p\,H=\lim_{p\to\infty}(1-a_p)\i_{\{X_T\geq 0\}}=R_{T,0}z\quad \Rightarrow \quad z=R_{0, T}\cdot \varphi_\infty H, 
\]
i.e., the entire endowment is kept in the bank account. 
This observation is consistent with the concept of total risk aversion, and it is in line with~\cite[][Lemma 14]{leukert_absicherungsstrategien_1999}. 
There, it is demonstrated that $\varphi_p H \rightarrow (H-a_\infty)_+$ for $p\rightarrow \infty$ almost surely and in $L^1(\P^*)$, for general (not necessarily constant) payoff functions $H = H(X_T)$.
\end{remark}

\begin{remark}
The knock-out feature of the digital European call that is present for risk-neutral/risk-taking investors ($p\in[0,1]$) makes hedging increasingly difficult if the underlying is close to the strike as maturity approaches, because the digital call's delta becomes unbounded. 
Appealingly, however, this knockout feature disappears for risk-averse investors ($p>1$), as one can see in Figure~\ref{fig:p1-5}, and the delta of these modified claims becomes more and more well-behaved as risk aversion increases ($p\rightarrow \infty$).
\end{remark}

\begin{figure}\label{fig:p1-5}
    \centering
    \includegraphics[scale=.56]{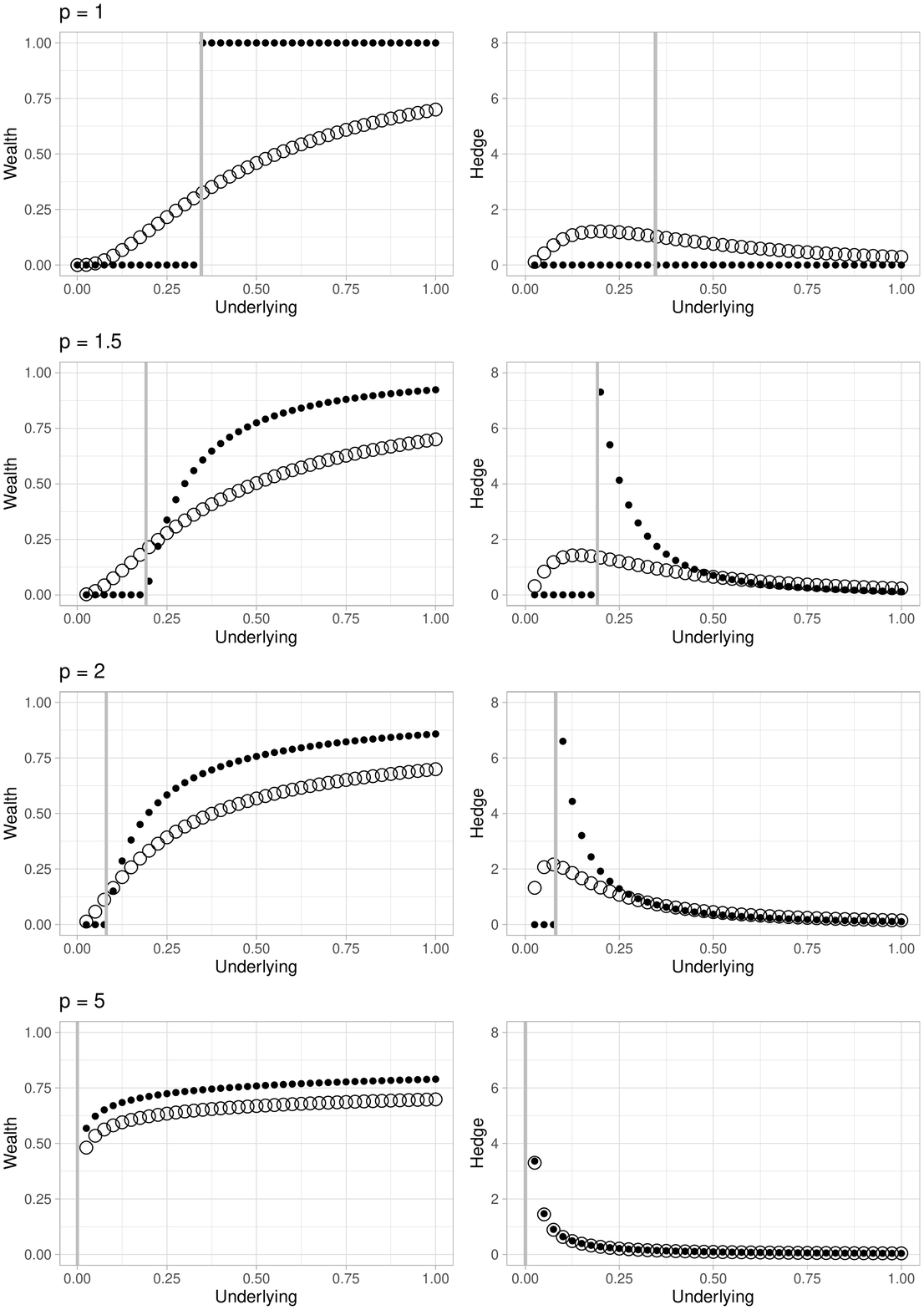}
    \caption{Wealth of an investor who seeks to  minimize expected shortfall (top) or lower partial moments of order $p$ relative to the investment goal $H$, respectively. The vertical lines denote the strike $L=L(p)$. Left column: wealth; right column: dollar hedge. Circles denote the initial state, while dots show terminal values. The maturity is $T=10$, the (annualized) drift $\mu=8\%$, the volatility $\sigma=30\%$, and the risk-free rate $r=1\%$. The investment goal is $H=1$, and the initial capital endowment is $z=0.7$. The required annualized return thus would be $(H/z)^{0.1}-1 \approx 3.6\% \gg r$.}
    \label{fig:my_label}
\end{figure}

\section{Downward Protection} \label{sec:down_prot}
The probability of reaching the target for the probability-maximizing policy, given by~\cite[][Theorem 3.1]{browne_reaching_1999} 
\begin{align}\label{eq:probsux}
    \sup_{f} \P_{(t, x)} [{X_T}^{(f)} \ge H ] = \Phi\left(\Phi^{-1}\left(\frac{x}{H_{t, T}}\right)+\sqrt{\pmb{\vartheta}^\top \pmb\vartheta \, (T-t)}\right), 
\end{align}
is the counter-probability of going bankrupt, 
which can be prohibitively high for practical purposes. 

\begin{remark}
If we assume an initial investment of two-thirds of the desired goal and a single risky asset with a drift of $6\%$, a volatility of $20\%$, and a zero risk-free rate, then the 'optimal' strategy entails a probability of losing everything of approximately $25\%$.  
\end{remark}

Clearly, this all-or-nothing strategy is too risky for most practical applications. 
Browne~\cite[][Section 8.2]{browne_reaching_1999} therefore proposed to control downside risk in the context of active portfolio management (cf. also~\cite{browne_beating_1999}). We adapt his approach to goal-based investing as follows. 

\begin{proposition}\label{prop:down_prot}
Consider an investor whose objective is to minimize the expected shortfall of her terminal wealth $V_T$ versus the goal $H \in \R_+$, with the additional requirement that the expected shortfall versus the discounted goal $H_{0, T}$ never exceed a predefined percentage $\delta\in[0, 1]$ of the latter. 
Then 
\begin{align*}
    &\sup_f \P\left[ X_T^{(f)} \ge H,\ \inf_{0\le s \le T} X_s^{(f)} \ge (1-\delta) H_{0,T} \,\bigg\rvert\, X_t = x \right] 
    \\ &= \Phi\left(\Phi^{-1}\left(\frac{x-(1-\delta) H_{t, T}}{\delta\, H_{t, T}}\right) + \sqrt{\pmb{\vartheta}^\top \pmb{\vartheta}\,(T-t)}\right). 
\end{align*}
\end{proposition}

\begin{corollary}[cf.~\cite{cvitanic_dynamic_1999}, Example 4.1]
Let $\varepsilon$ be a given positive real number. 
It follows from Proposition \ref{prop:down_prot} that the smallest initial endowment $x_\varepsilon > 0$ required so that the probability of violating the shortfall constraint is bounded from above by $\varepsilon$ is given by
$$
x_\varepsilon = \left[ \Phi\left( \Phi^{(-1)}(1-\varepsilon)-\sqrt{\pmb{\vartheta}^\top \pmb{\vartheta} \,T} \right) + 1-\delta \right] H_{0, T} .
$$
\end{corollary}

Note that, as $\varepsilon \rightarrow 1$, the initial endowment $x_\varepsilon$ tends to the discounted goal $H_{0, T}$ minus the shortfall allowance $\delta \,H_{0, T}$. 

\subsection{The nature of the claim with downward protection}
If maximizing the probability of reaching an investment goal is equivalent to replicating a digital European call option (cf.~\cite[][Proposition 4.1]{browne_reaching_1999}), what interpretation can be given to the situation in this section? 
\\ 
First, let us rephrase the optimal policy, given for the general case in~\cite[][Theorem 3.1]{browne_reaching_1999}, for constant coefficients and in the presence of a downward risk limit:
\begin{align*}
    f^*_t(x-(1-\delta)\,H_{t, T}; \delta \,H) = \frac{\delta \,H_{t, T}}{\sigma \sqrt{T-t}}\phi\left( \Phi^{-1}\left(\frac{x-(1-\delta) \,H_{t, T}}{\delta \,H_{t, T}}\right) \right). 
\end{align*}
Now, if we evaluate $f^*_t$ at $x = C(t, X_t;\delta)$, where 
\begin{align*}
    C(t, X_t;\delta) = \delta \,H_{t, T}\,\Phi\left( \frac{\log \frac{X_t}{K^*} + (r-\frac{\sigma^2}{2})(T-t)}{\sigma\,\sqrt{T-t}} \right) + (1-\delta) \,H_{t, T}
\end{align*} 
then 
\begin{align*}
    f_t^*(C(t, X_t;\delta)-(1-\delta) \,H_{t, T}; \delta \,H) &= \frac{\delta \,H_{t, T}}{\sigma\,\sqrt{T-t}}\, \phi\left( \Phi^{-1}\left(\frac{C(t, X_t;\delta)-(1-\delta) \,H_{r, T}}{\delta \,H_{t, T}}\right) \right)
    \\
    &= \frac{\delta \,H_{t, T}}{\sigma\,\sqrt{T-t}}\,\phi\left( \frac{\log \frac{X_t}{K^*}+\left(r-\frac{\sigma^2}{2}\right)(T-t))}{\sigma \,\sqrt{T-t}} \right) 
    \\ 
    &= \Delta_t \cdot X_t
\end{align*}
where $\Delta_t$ is the delta of the digital European call option paying $\delta \,H$ at maturity if $X_T \ge K^*$, and nothing otherwise. 
The optimal policy thus consists of initially investing $(1-\delta) \,H_{0, T}$ into a bond, and the remainder into a digital European call option with said characteristics. 
As before, the strike $K^*$ of this contingent claim depends implicitly on the initial endowment $z$. 

\section{Deep Hedging} \label{sec:deep_hedg}
The investment strategies derived in the previous sections cannot be transferred to more realistic settings without further ado. The optimality fundamentally relies on the completeness of the financial market model as well as the simplistic distributional assumption on the price dynamics. More sophisticated price dynamics, for instance involving rough volatility, inevitably lead to incomplete market models. Furthermore, minimizing lower partial moments in such intricate environments may hardly be analytically tractable. It remains unclear whether the duality principle between the optimization problem and the hedging of a qualitatively similar payoff prevails. In contrast, simply applying the proposed delta hedging strategies for different price dynamics can be arbitrarily bad. Another impediment for applications in the real world are discrete hedging schedules and transaction cost. Therefore, we investigate whether we manage to circumvent these delicate issues by applying the striking approach of Deep Hedging as proposed in \cite{buehler_deep_2019}. Subsequently, we present our findings for the one-dimensional case.

For any $t\in\{0,1,2,\hdots,N\}$ in some discrete time grid with horizon $N\in\N$, we consider a feedforward neural network
\[F_t=\left(\phi\circ A_t^{(2)}\right)\circ\left(\phi\circ A_t^{(1)}\right)\circ\left(\phi\circ A_t^{(0)}\right)\]
with some affine functions
\[A_t^{(0)}:\R^2\longrightarrow\R^{10},\ A_t^{(1)}:\R^{10}\longrightarrow\R^{10},\ A_t^{(2)}:\R^{10}\longrightarrow\R\]
and the sigmoid activation function $\phi(x)=(1+e^{-x})^{-1}$. The input layer consists of the current holding $\xi_{t\leli}$ before rehedging and the moneyness $S_t/S_0$, where $S_t$ is the marginal distribution of a geometric Brownian motion as considered above. The output layer reveals the outcome $\xi_t$ of the re-hedging at the time instance $t$. Similarly as above, we aim at optimizing a functional of the terminal wealth $V_T$ that can be derived iteratively. Let $b_{0\leli}$ denote the initial holdings in the bank account bearing the risk-free rate $r\in\R$, $\xi_{0\leli}$ denote the initial holdings in the underlying, $\kappa\geq 0$ the coefficient for proportional transaction cost, and $\tau>0$ the year fraction of a time step. Hence, the value of the portfolio before and after re-hedging at time $0$ are given by
\begin{align*}
V_{0\leli}&=b_{0\leli}+\xi_{0\leli}S_0,\\
V_0&=b_0+\xi_0S_0,
\intertext{where $b_0:=b_{0\leli}-\left(\xi_0-\xi_{0\leli}\right)S_0-\kappa\left|\xi_0-\xi_{0\leli}\right|S_0$ satisfied the self-financing principle. Then, we proceed consistently in terms of the iteration}
V_{t\leli}&=b_{t\leli}+\xi_{t\leli}S_t\\
V_t&=b_t+\xi_tS_t,
\intertext{where $b_{t\leli}=b_{t-1}e^{r\tau}$, $\xi_{t\leli}=\xi_{t-1}$ and $b_t=b_{t\leli}-\left(\xi_t-\xi_{t\leli}\right)S_t-\kappa\left|\xi_t-\xi_{t\leli}\right|S_t$ for $t\in\{1,2,\hdots,N-1\}$. At maturity, we have to bear additionally the unwinding cost. Hence,}
V_T&=b_{T-1}e^{r\tau}+\xi_{T-1}S_T-\kappa\left|\xi_{T-1}\right|S_T.
\end{align*}
For experimental purposes, we chose similar parameters as in Figure~\ref{fig:p1-5}; a maturity $T=10$, a discretization $N=52T$ (i.e., weekly rehedging with $\tau=1/52$), a risk-free rate $r=1\%$, a drift $\mu=8\%$ and a volatility $\sigma=30\%$. The initial state of the market and the wealth are standardized to $S_0=100$, $b_{0\leli}=70$ and 
$\xi_{0\leli}=0$. The ultimate goal is to reach the deterministic payoff $H=100$;
this refers to as a continuously compounded return of $h\approx 3.6\%$. Let $J\in\N$ be a sufficiently large number of simulated paths $S^{(j)}=(S_t^{(j)})_{t=0,1,2,\hdots,N}$, e.g., $J=10^4$. Given this parameter set, we seek to find optimal rehedging strategies. This can be achieved by applying a suitable backpropagation algorithm on the deep neural network architecture that consolidates the above feedforward neural network instances together with the intermediary accounting routines. A direct translation of the above concept is the minimization of the loss
\[\frac{1}{J}\sum_{j=1}^J\max\left\{H-V_T^{(j)},0\right\}^p.\]
We modify the loss function for two crucial reasons. Firstly, the function $x\mapsto\max\{H-x,0\}$ is non-differentiable at the point $H$ and ignores any points beyond $H$. This raises concerns on the stability of the learning algorithm. Therefore, we replace the maximum with the softplus function $\log{(1+e^{x})}$. Secondly, the natural extension of the loss function apparently has an undesirable local minimum for strategies with a deterministic equity portion $\xi_t\equiv\xi\in[0,1]$; see Figure~\ref{fig:convexity} above.
\begin{figure}
\centering
\includegraphics[width=0.5\textwidth]{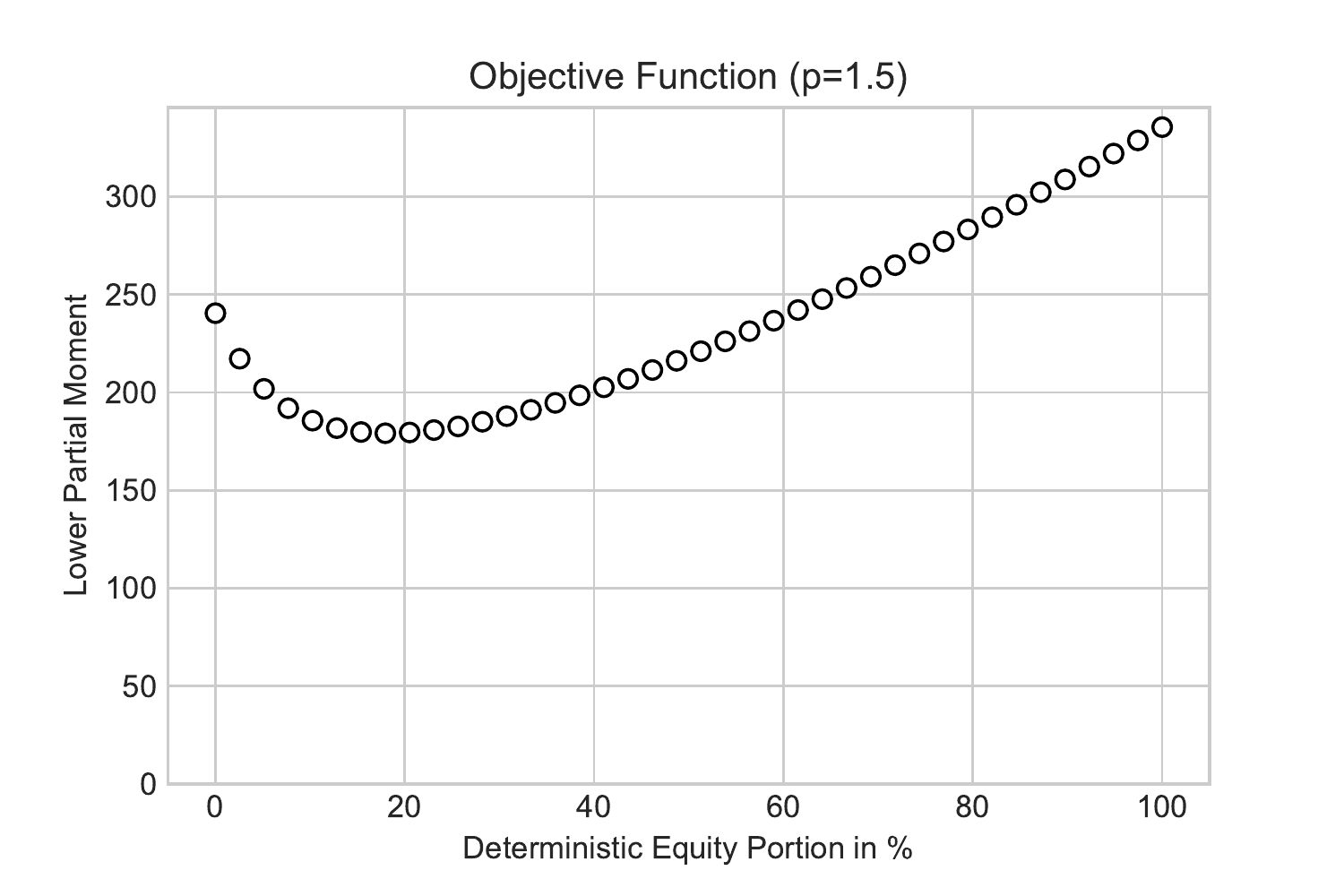}
\caption{This chart exhibits the local minimum of the lower partial moments in the neighborhood of static strategies.}
\label{fig:convexity}
\end{figure}
Without further interventions, the learning algorithms often gets stuck in the suboptimal neighborhood of static strategies. Therefore, we also penalize deviances beyond $H$ in terms of
\[\frac{1}{J}\sum_{j=1}^J\left(\log\left\{1+\exp\left\{H-V_T^{(j)}\right\}\right\}\right)^p+\lambda\log\left\{1+\exp\left\{V_T^{(j)}-H\right\}\right\}.\]
for a regularization parameter $\lambda=0.1$. It needs to be noted that the introduction of the positive second summand does not alter the global optimum. The following charts exhibit the out-of-sample performance of a trained artificial financial agent for $p\in\{1,1.5,5\}$ and $\kappa\in\{0,0.005\}$. All charts are generated with the same sample data. The training phase of the Jupyter notebook takes in each case approximately $2.5h$ on Google Colab. As a benchmark, we also show the performance of naively applying the continuous time optimal hedging strategy on the same weekly time grid. 

For $p\in\{1,1.5\}$, deep hedging mitigates the risk of large losses. In the absence of transaction cost, Deep Hedging is not able to surpass the benchmark consistently, at least not for the selected parameters and without further measures. However, the strength of Deep Hedging is particularly evident in the presence of transaction costs. Likewise, it could be extended accordingly to more realistic dynamics of the underlying for which analytical solutions are typically not available. 
The empirically derived expected terminal wealth and the Value-at-Risk to a significance of $5\%$ for the different investment strategies are lined up in Table~\ref{tbl:lineup}. 
Remarkably, due to accounting for offsetting effects of an adjusted hedge and borne transaction cost, Deep Hedging leads to an improved Value-at-Risk in the presence of transaction cost.

\begin{figure}
\textbf{Final Wealth Distribution without Transaction Cost}\\
\vspace{1em}
\begin{center}
Deep Hedging \hfill Discrete Delta Hedging
\end{center}
$p=1$\\
\includegraphics[width=0.5\textwidth]{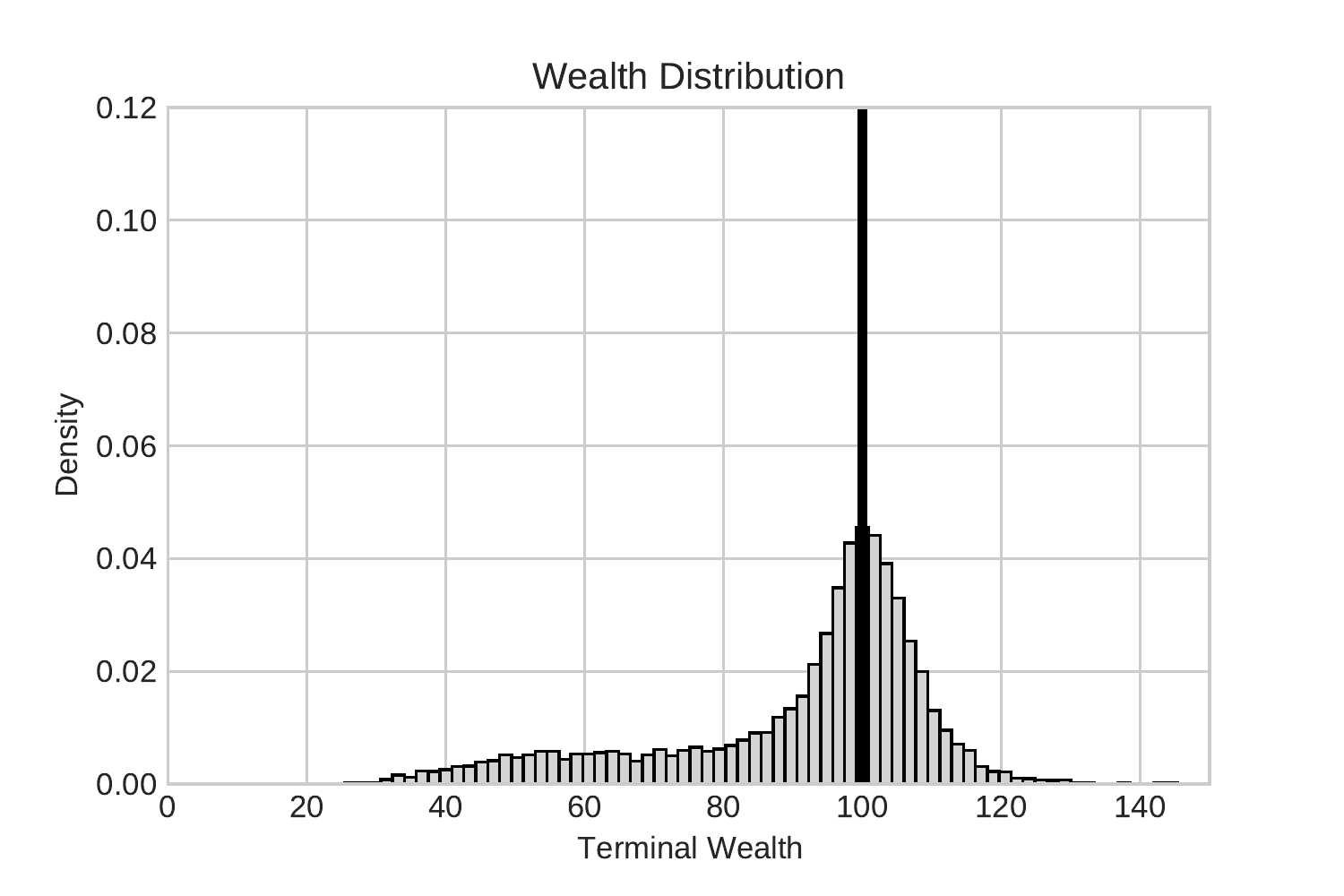}\includegraphics[width=0.5\textwidth]{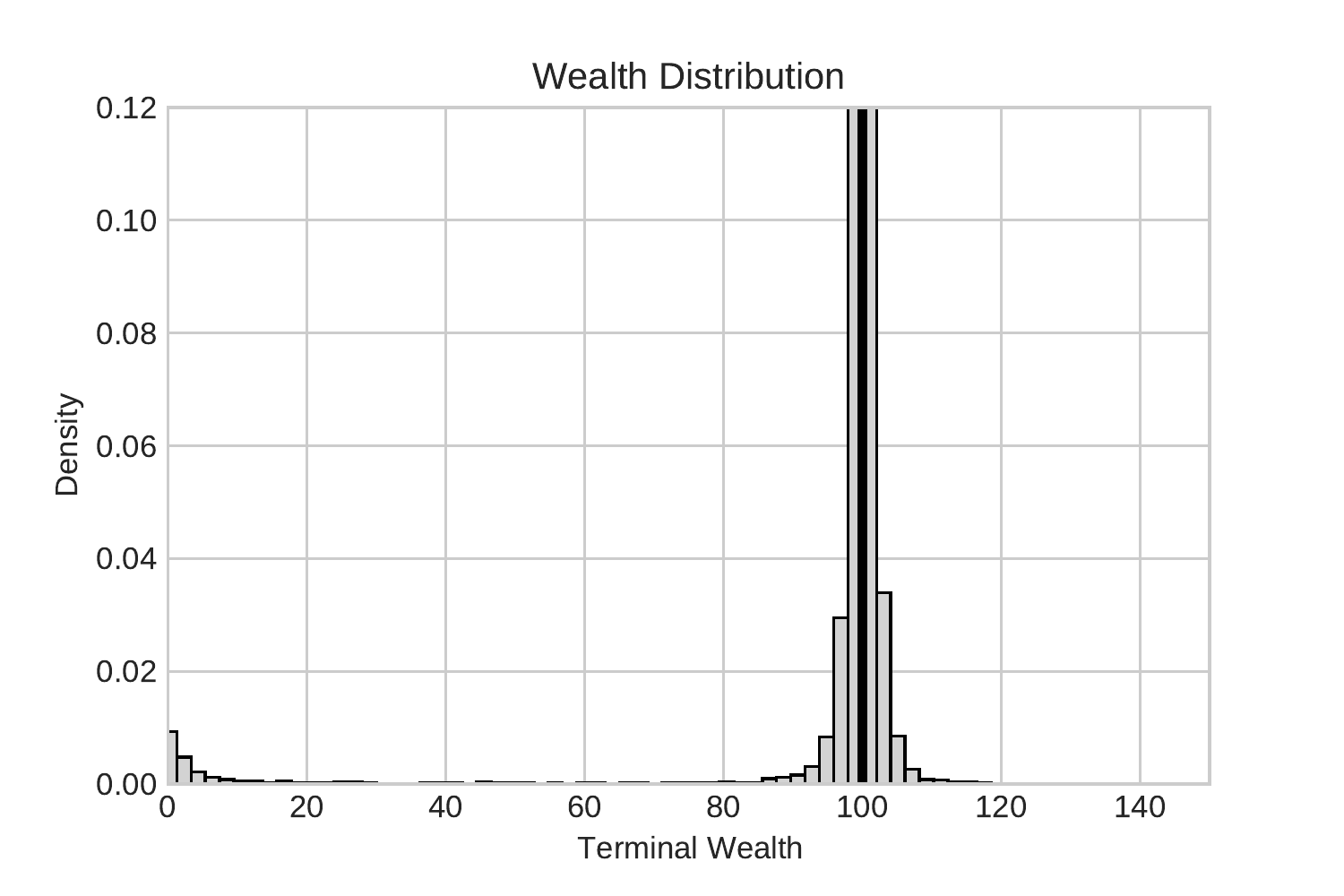}
$p=1.5$\\
\includegraphics[width=0.5\textwidth]{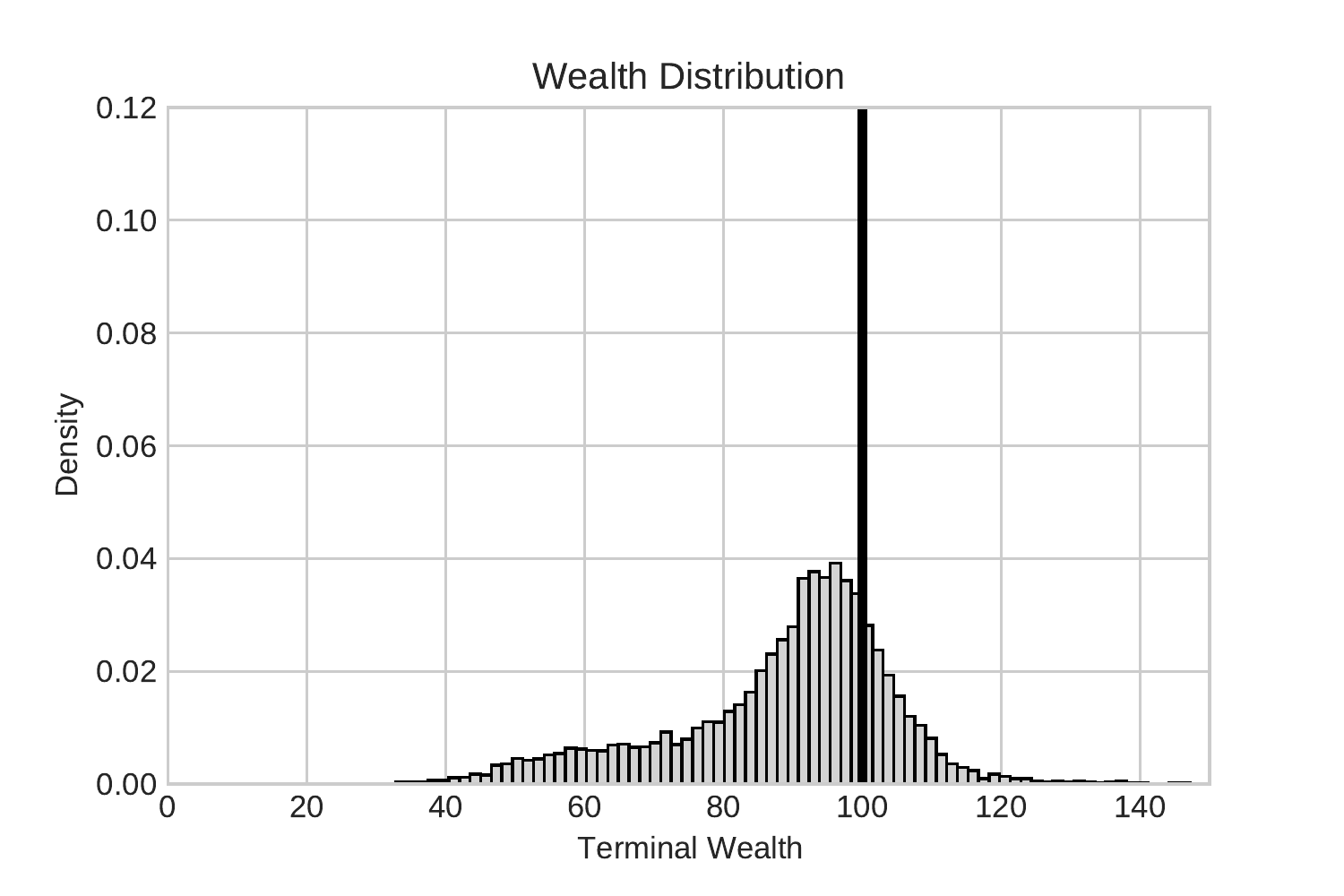}\includegraphics[width=0.5\textwidth]{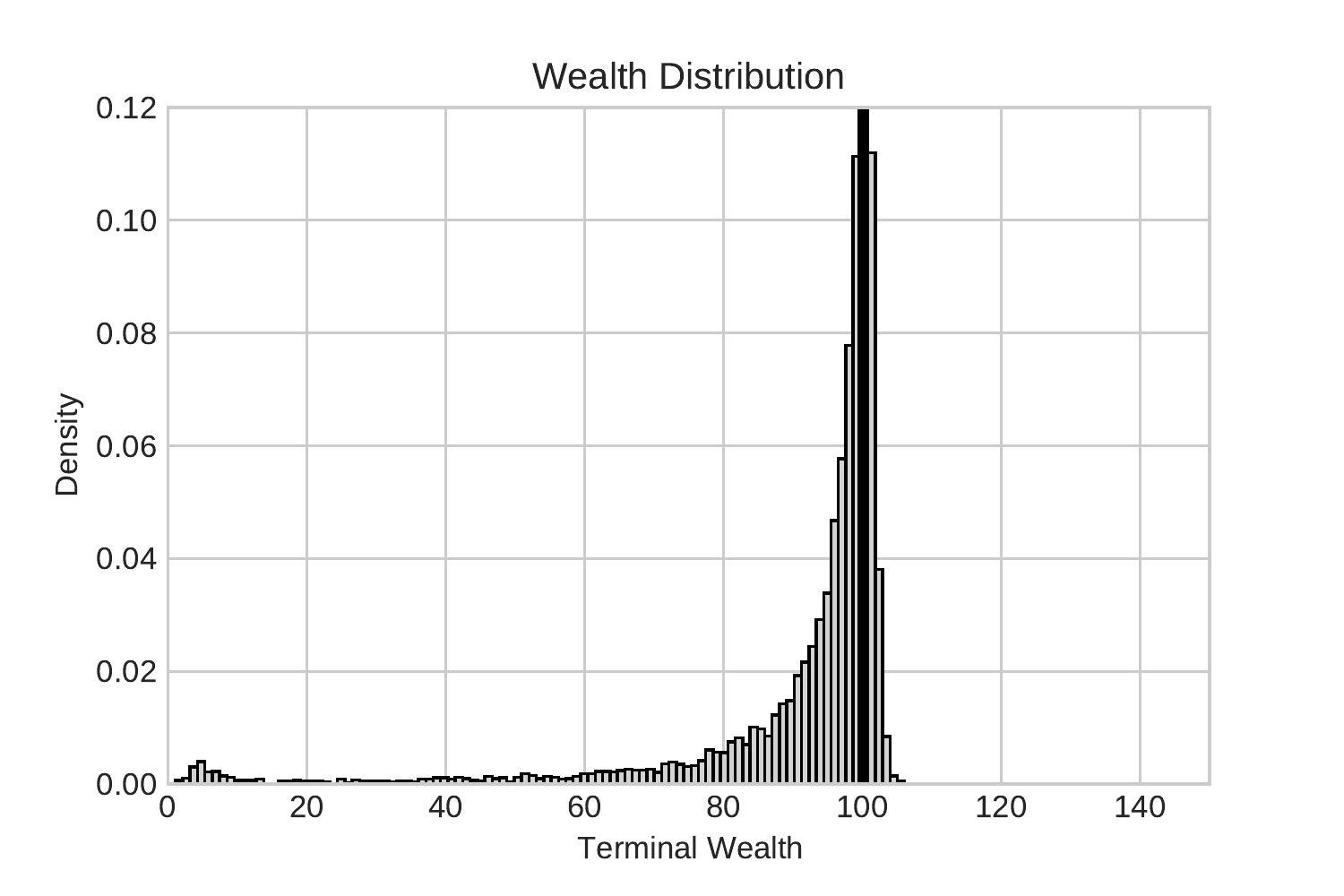}
$p=5$\\
\includegraphics[width=0.5\textwidth]{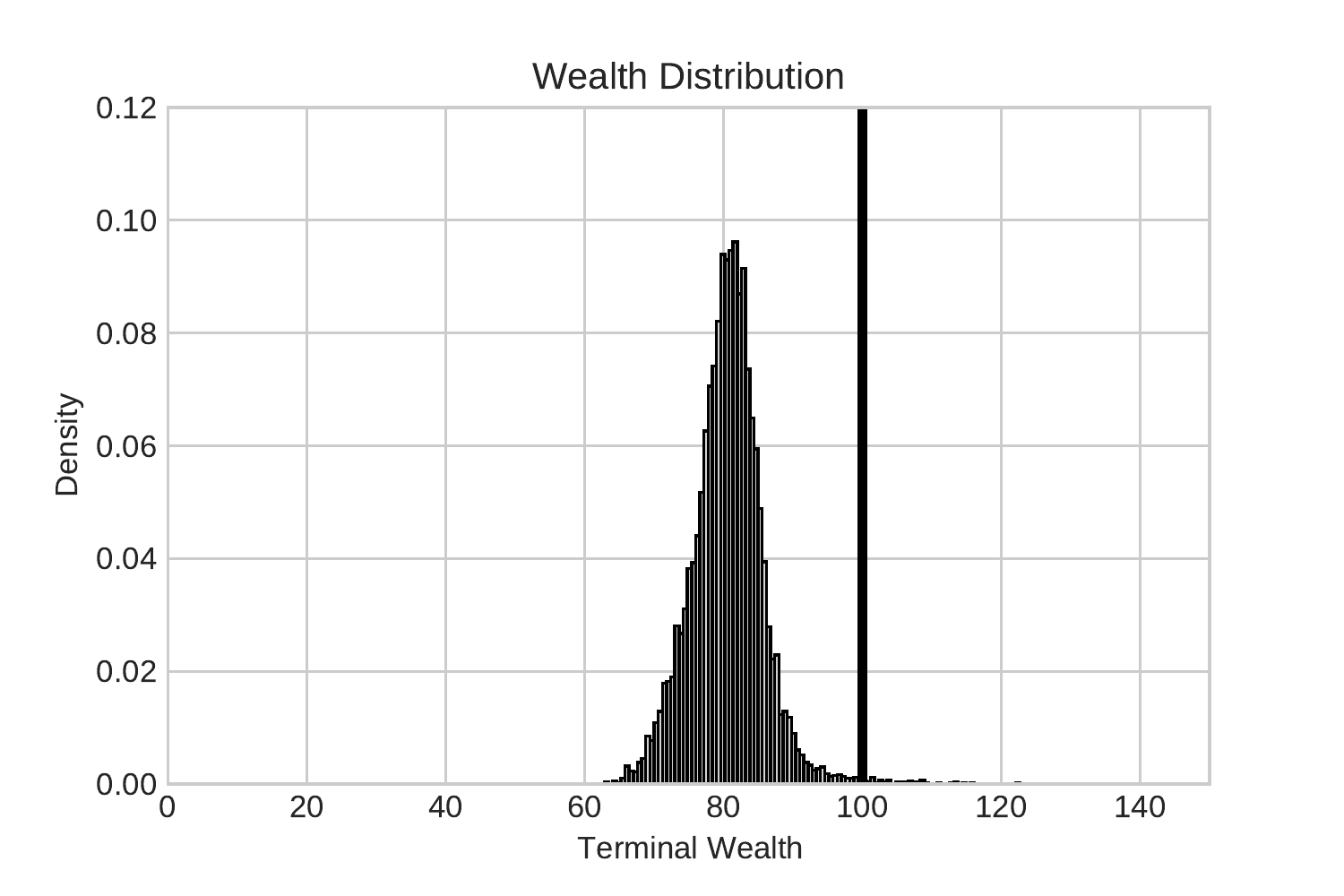}\includegraphics[width=0.5\textwidth]{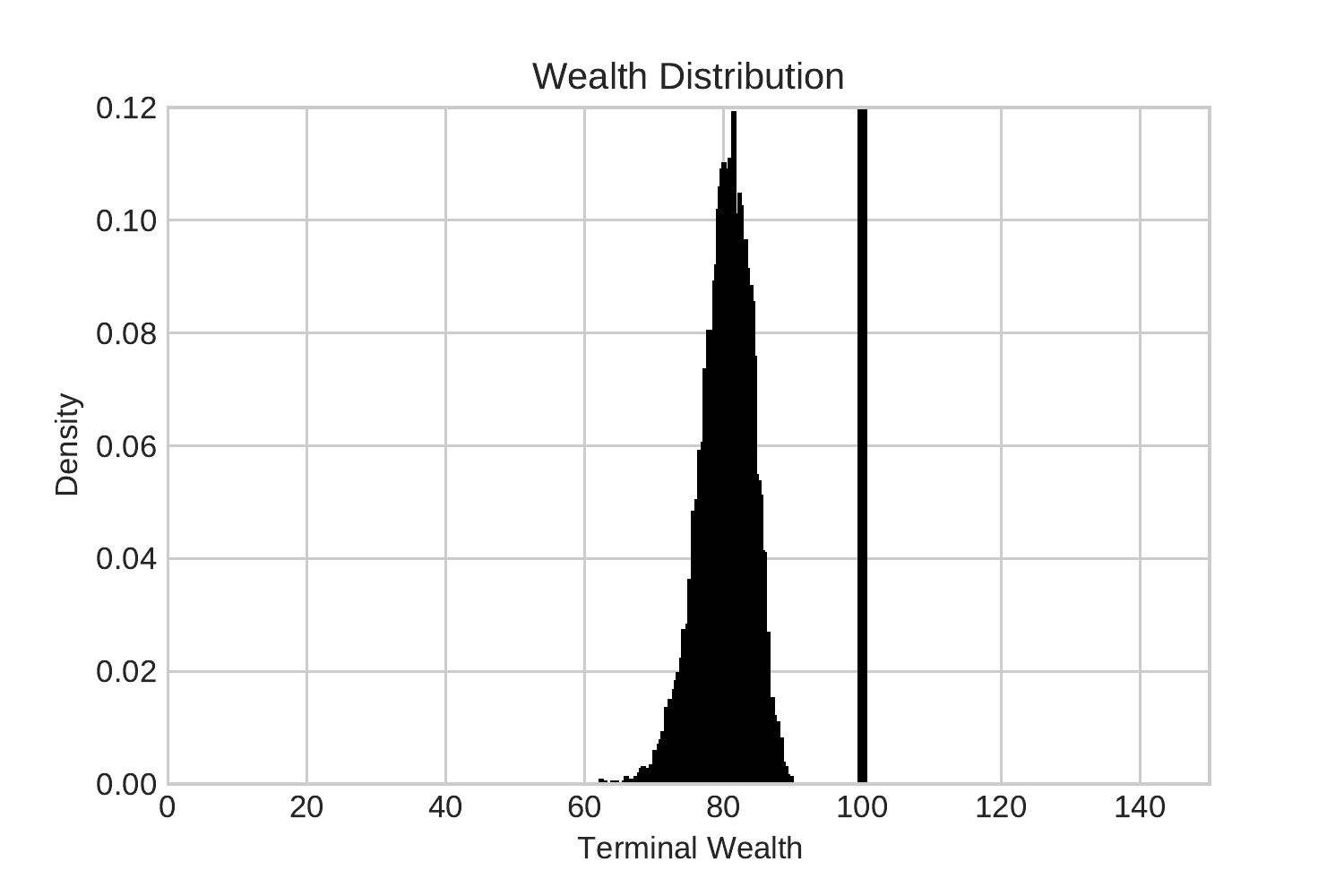}
\caption{For different choices of the risk aversion $p$, the empirical probability density function of the final payoffs depict the performance of a trained artificial financial agent in the absence of transaction cost in the left column compared to naively applying the corresponding continuous time optimal delta hedging strategy in the right column. The solid line represents the primary target payoff.}
\end{figure}

\begin{figure}
\textbf{Payoff Diagram without Transaction Cost}\\
\vspace{1em}
\begin{center}
Deep Hedging \hfill Discrete Delta Hedging
\end{center}
$p=1$\\
\includegraphics[width=0.5\textwidth]{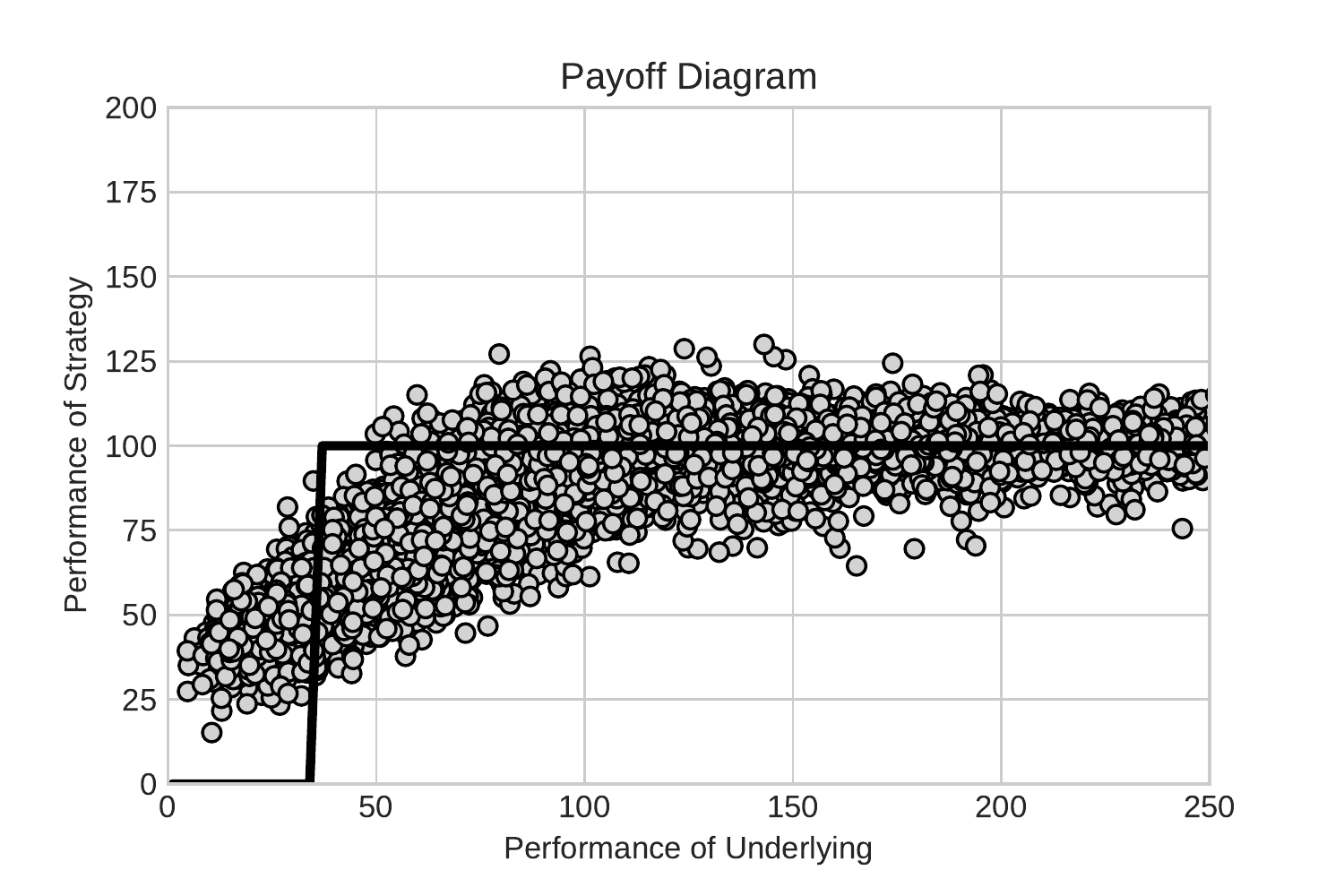}\includegraphics[width=0.5\textwidth]{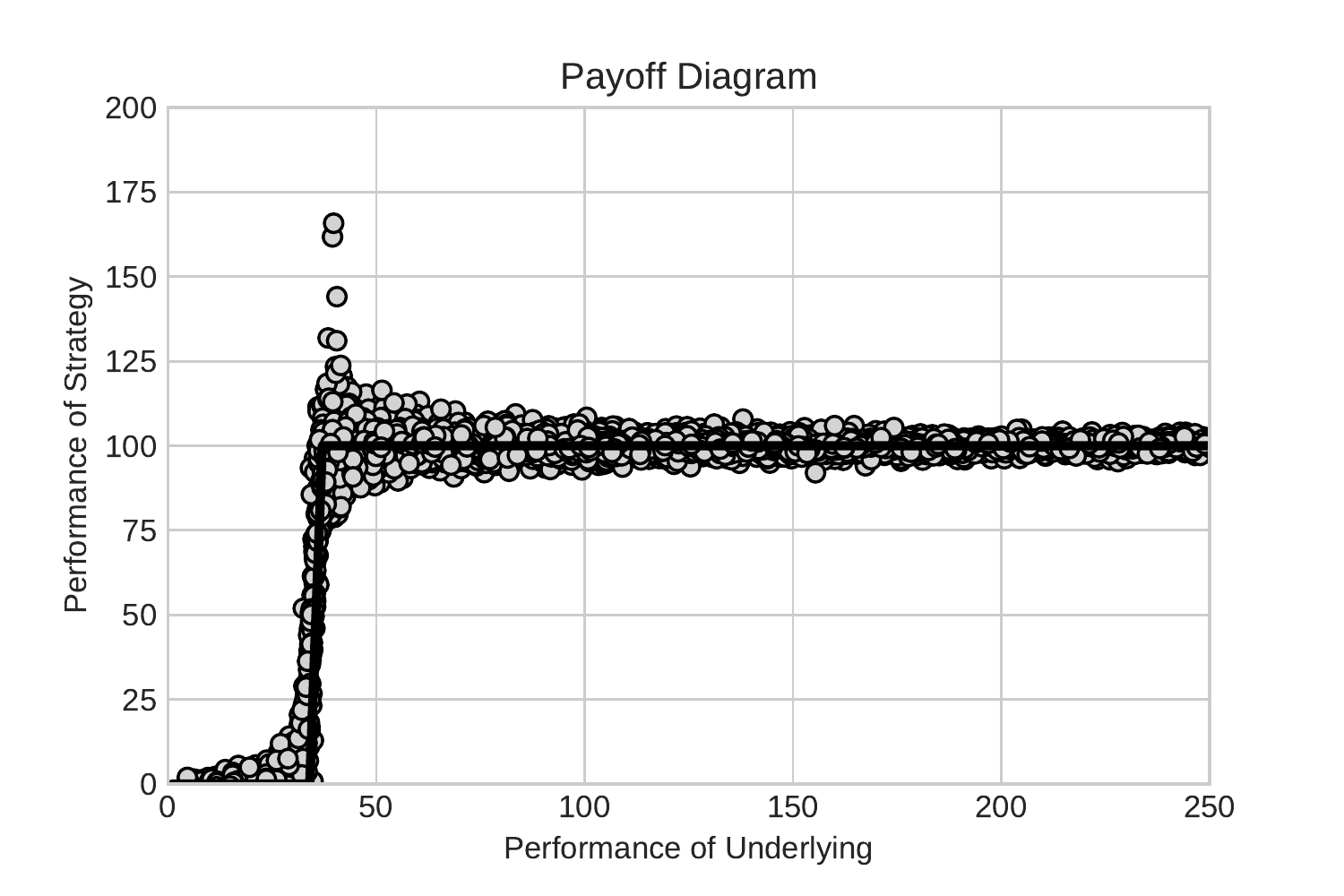}
$p=1.5$\\
\includegraphics[width=0.5\textwidth]{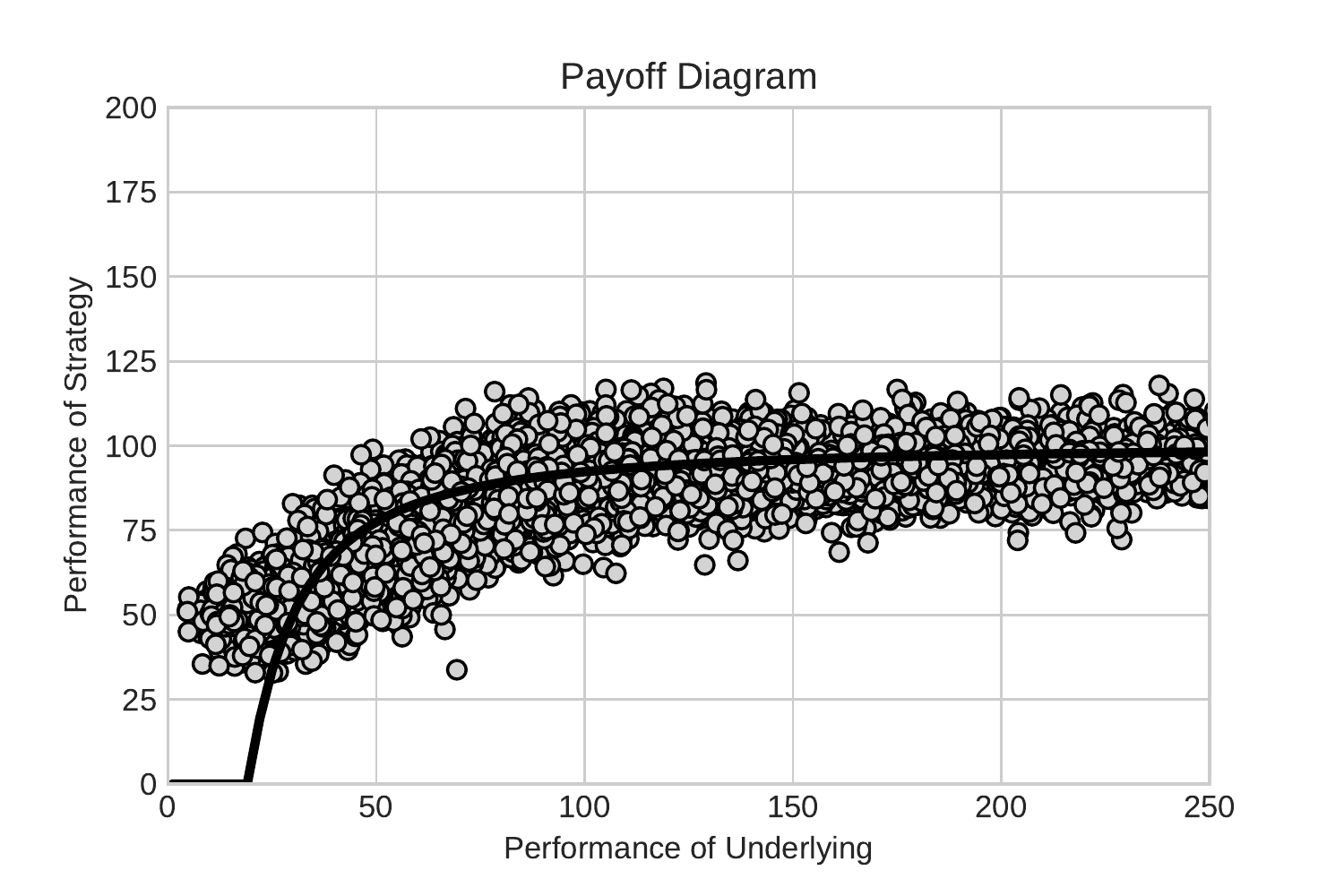}\includegraphics[width=0.5\textwidth]{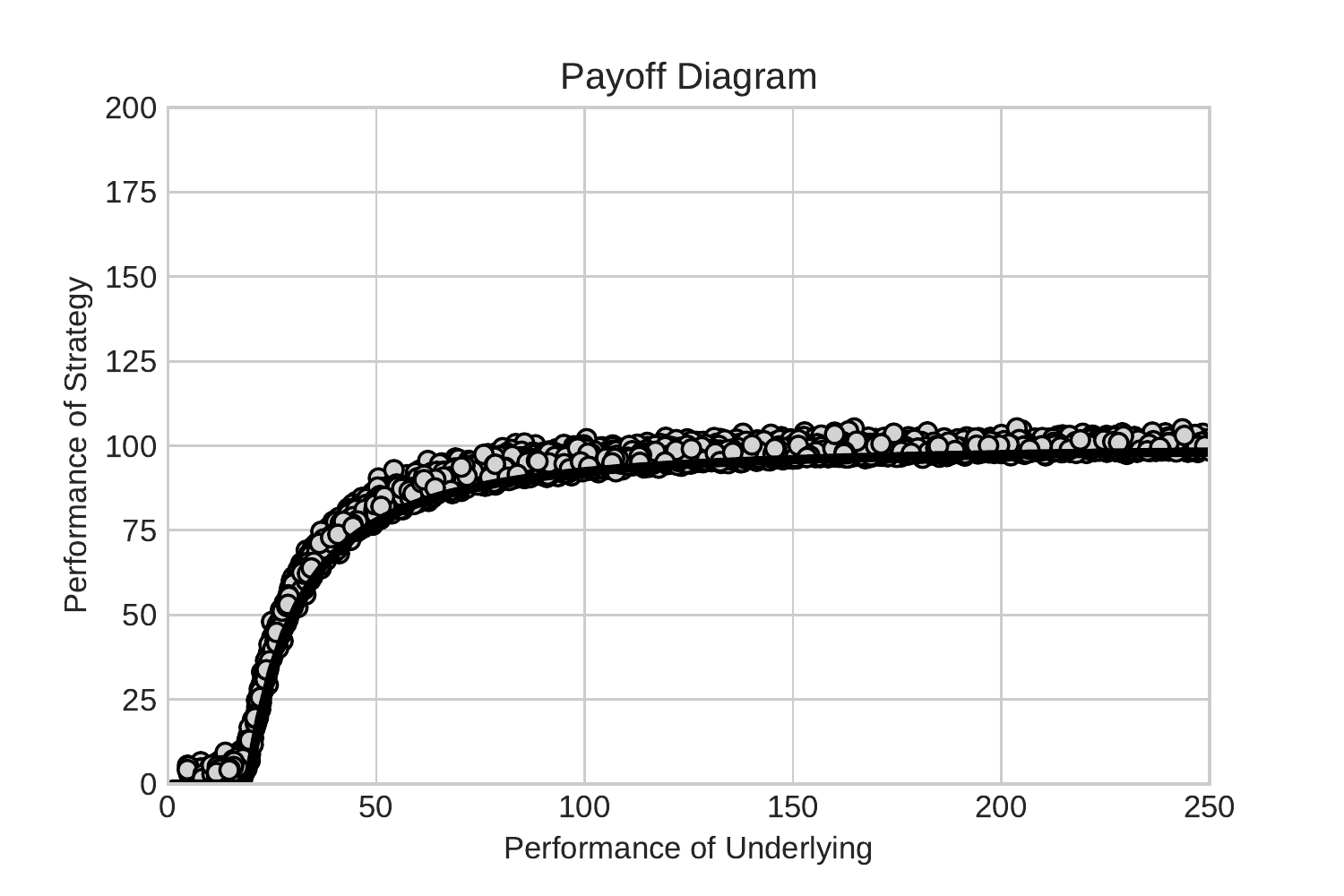}
$p=5$\\
\includegraphics[width=0.5\textwidth]{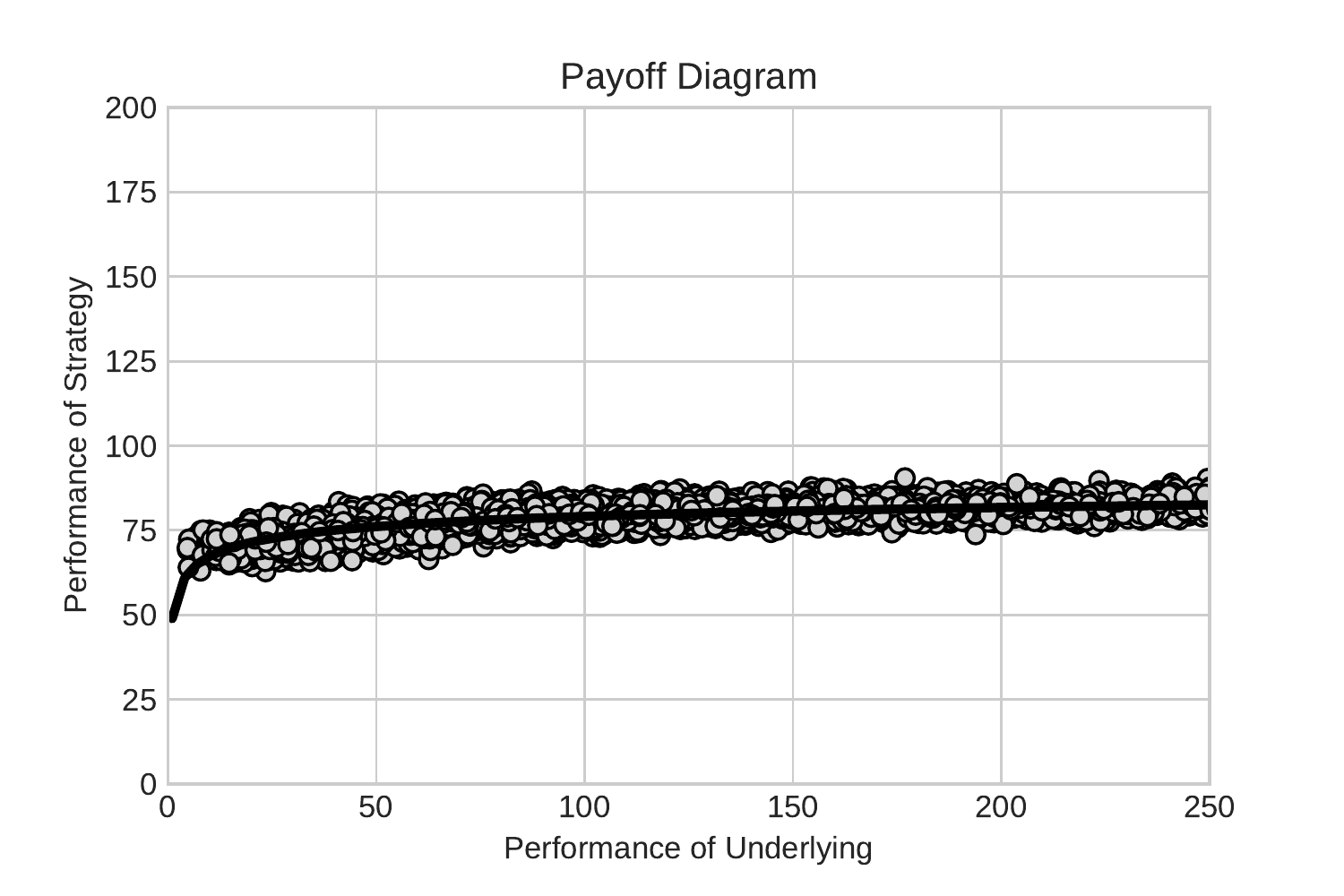}\includegraphics[width=0.5\textwidth]{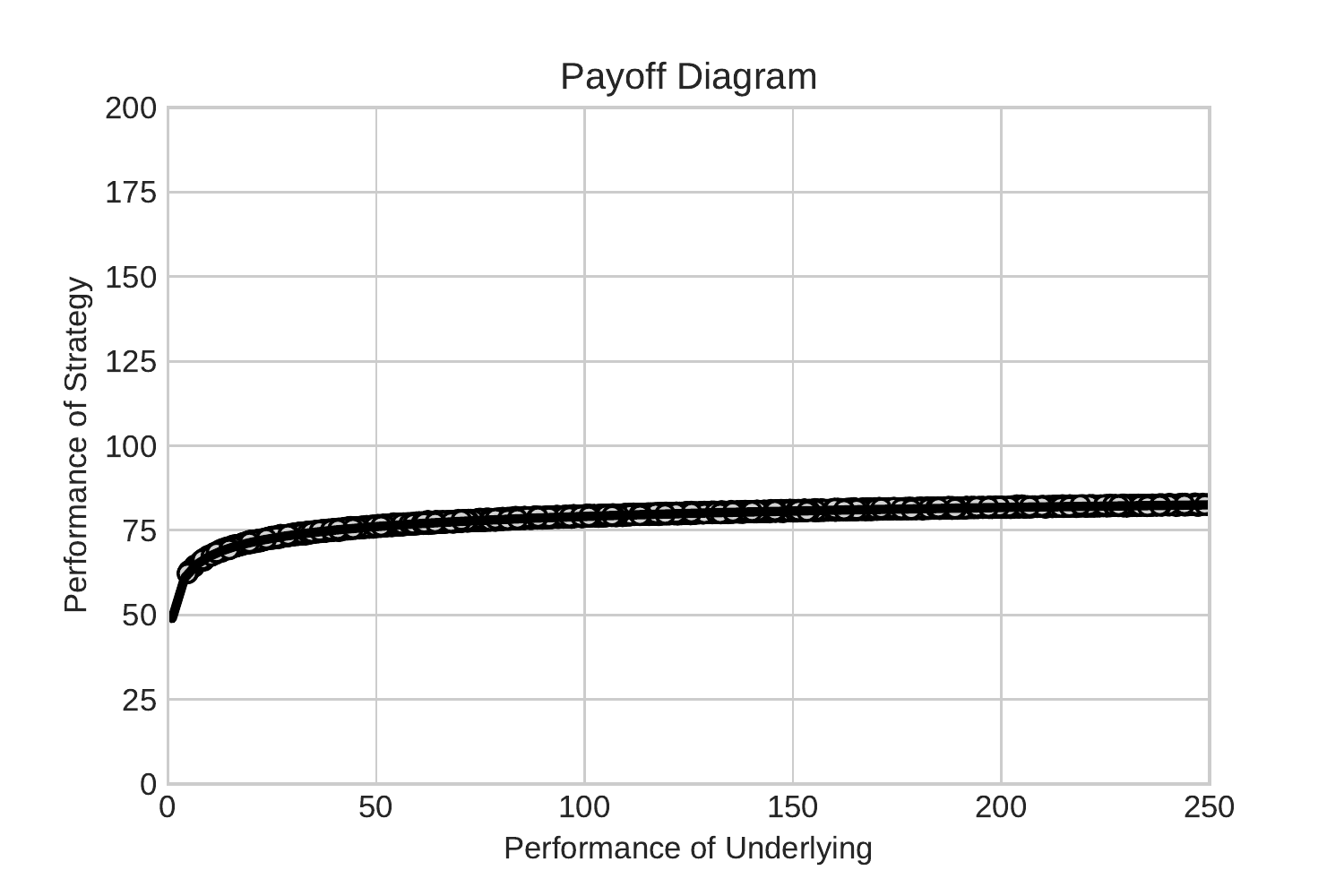}
\caption{For different choices of the risk aversion $p$, the scatter plots depict the final payoffs depending on the performance of the underlying for a trained artificial financial agent in the absence of transaction cost in the left column compared to naively applying the corresponding continuous time optimal delta hedging strategy in the right column. The solid line represents the secondary target payoff originating from the duality result of the continuous-time problem.}
\end{figure}

\begin{figure}
\textbf{Final Wealth Distribution with Transaction Cost}\\
\vspace{1em}
\begin{center}
Deep Hedging \hfill Discrete Delta Hedging
\end{center}
$p=1$\\
\includegraphics[width=0.5\textwidth]{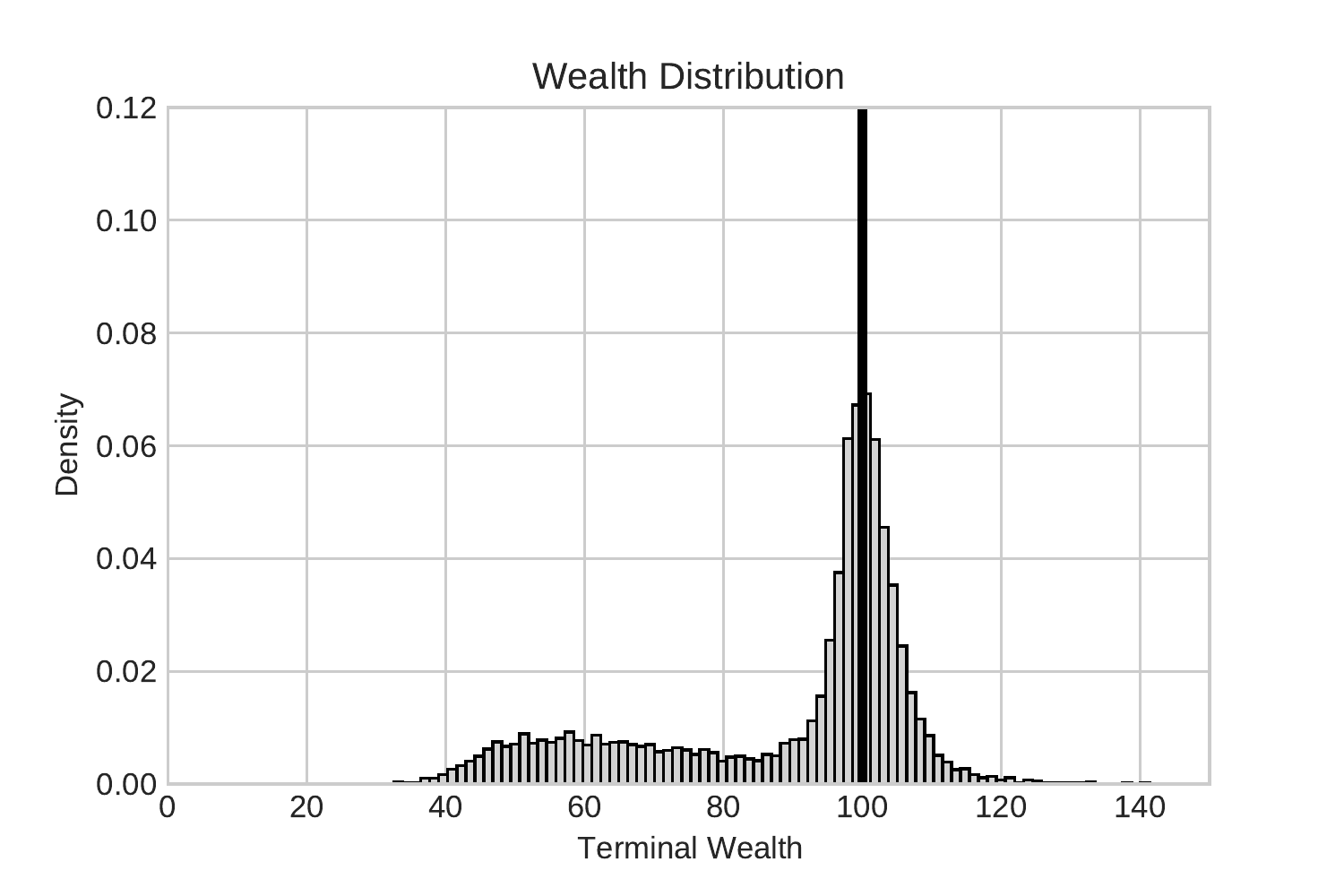}\includegraphics[width=0.5\textwidth]{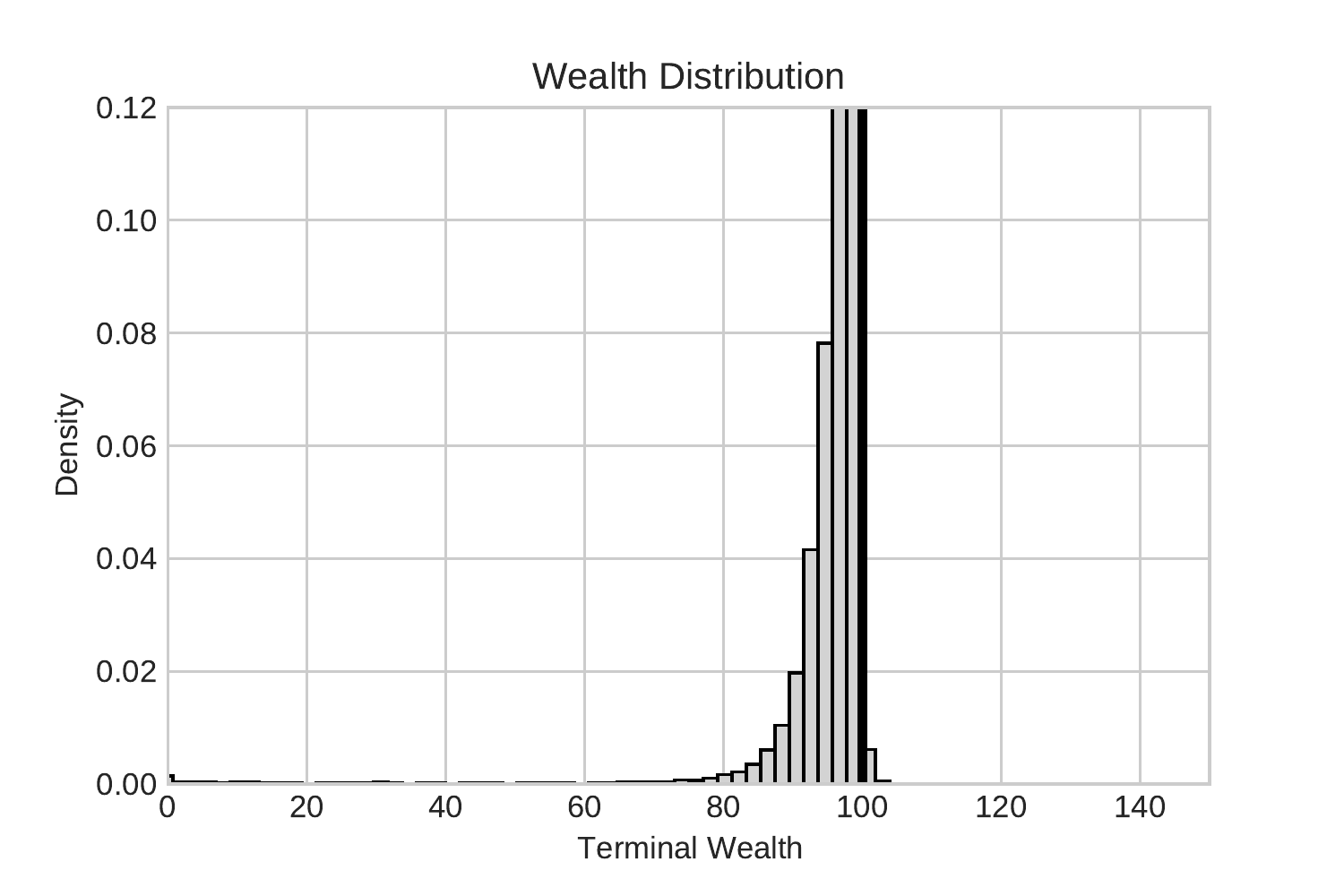}
$p=1.5$\\
\includegraphics[width=0.5\textwidth]{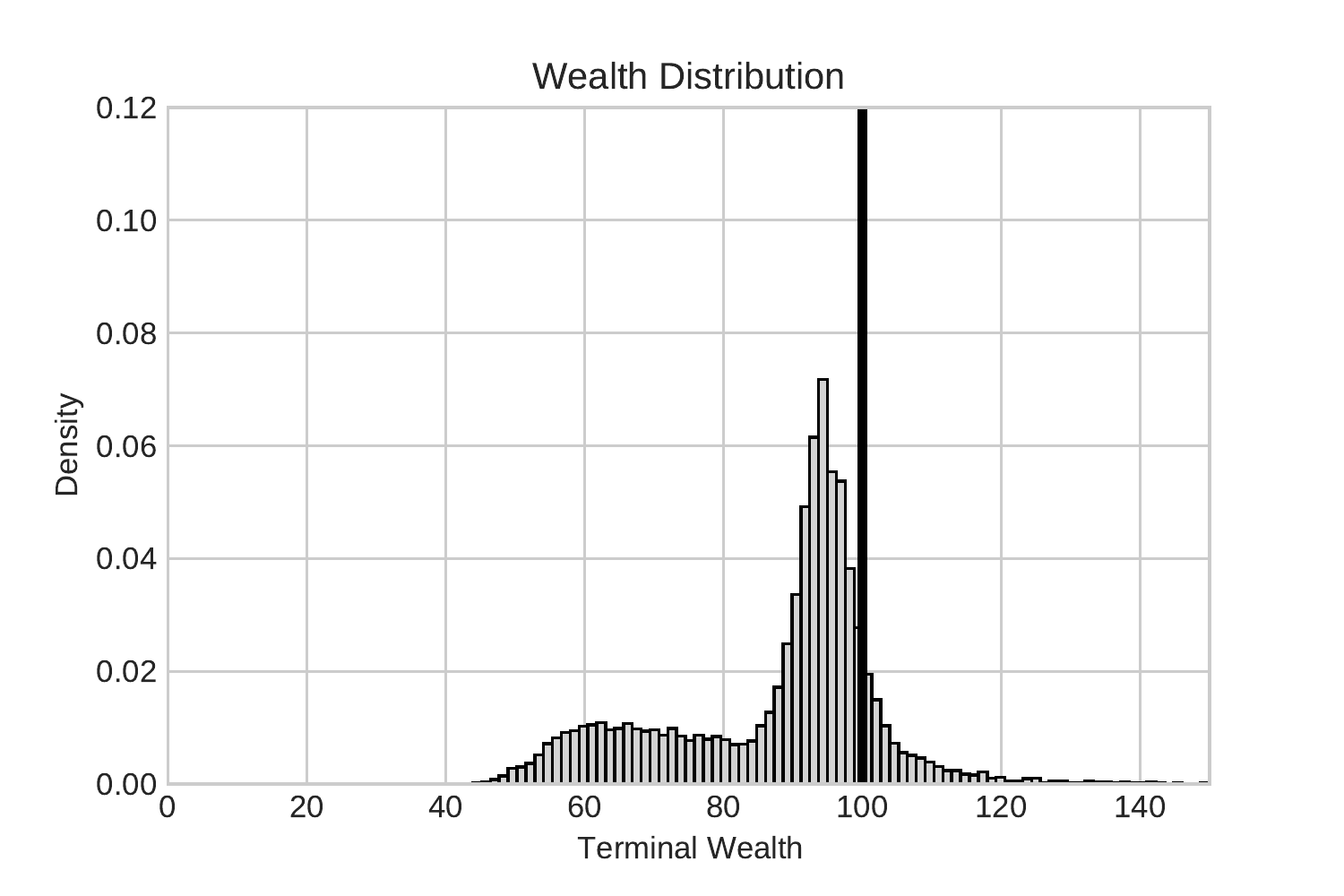}\includegraphics[width=0.5\textwidth]{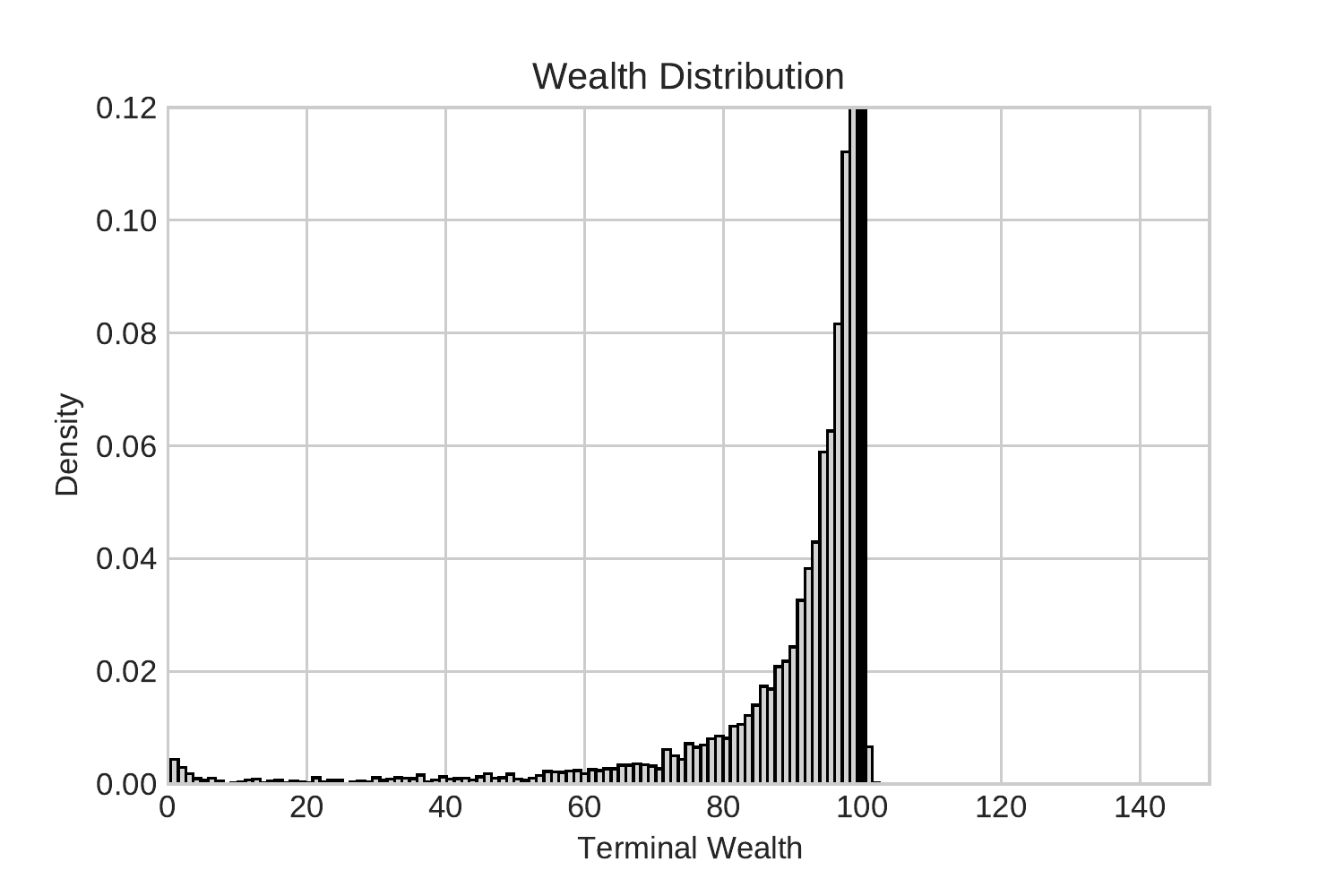}
$p=5$\\
\includegraphics[width=0.5\textwidth]{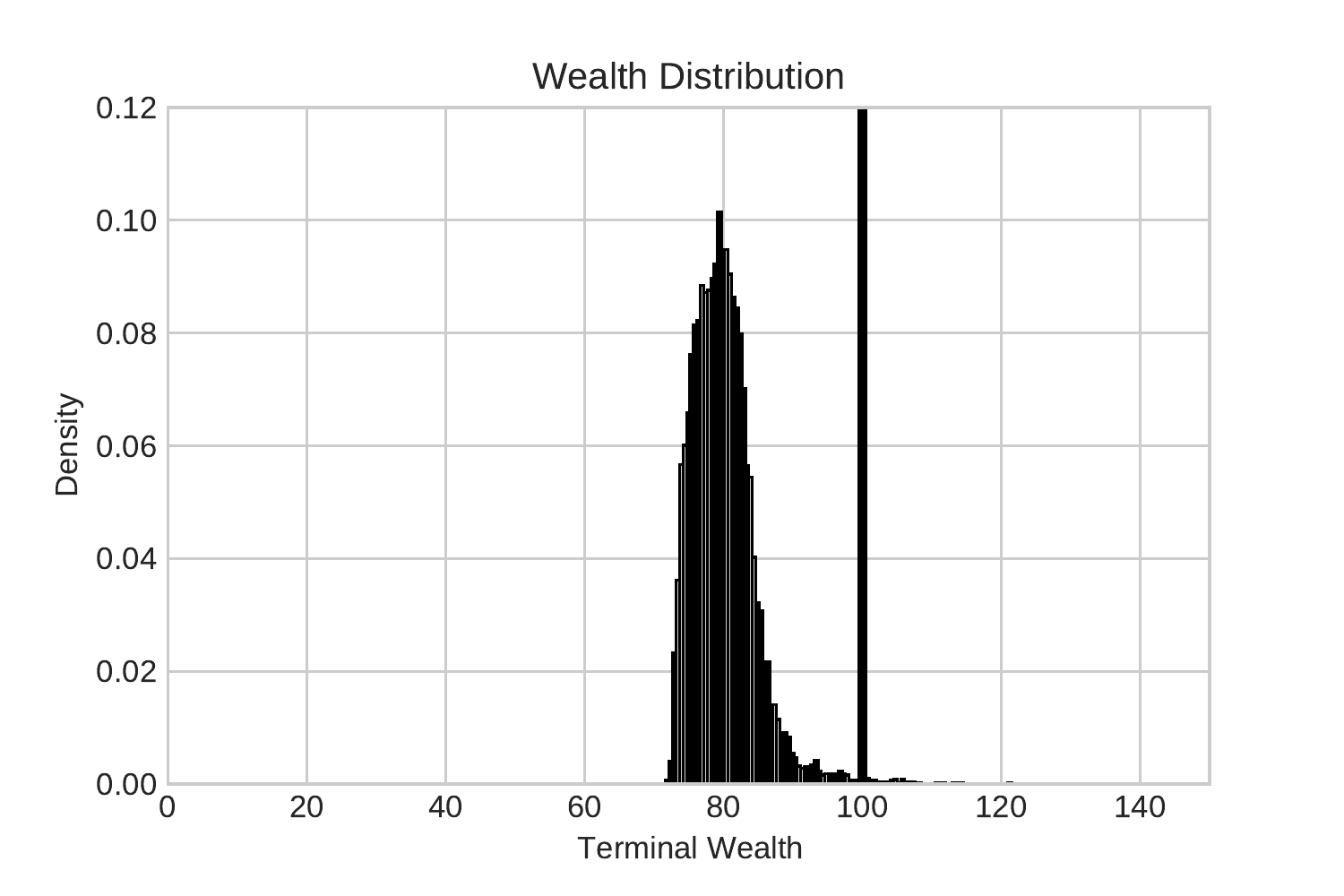}\includegraphics[width=0.5\textwidth]{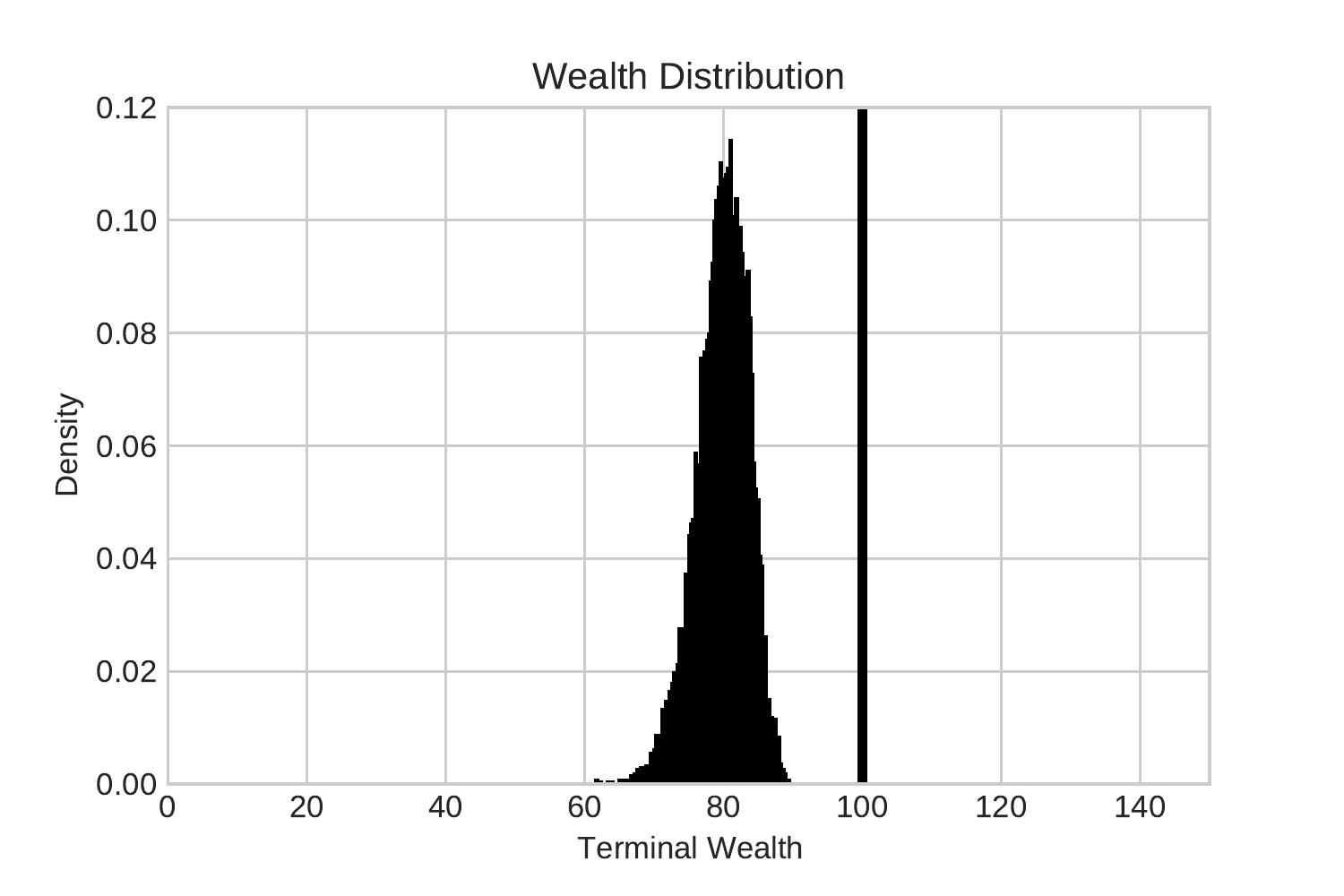}
\caption{For different choices of the risk aversion $p$, the empirical probability density function of the final payoffs depict the performance of a trained artificial financial agent in the presence of proportional transaction cost in the left column compared to naively applying the corresponding continuous time optimal delta hedging strategy in the right column. The solid line represents the primary target payoff.}
\end{figure}

\begin{figure}
\textbf{Payoff Diagram with Transaction Cost}\\
\vspace{1em}
\begin{center}
Deep Hedging \hfill Discrete Delta Hedging
\end{center}
$p=1$\\
\includegraphics[width=0.5\textwidth]{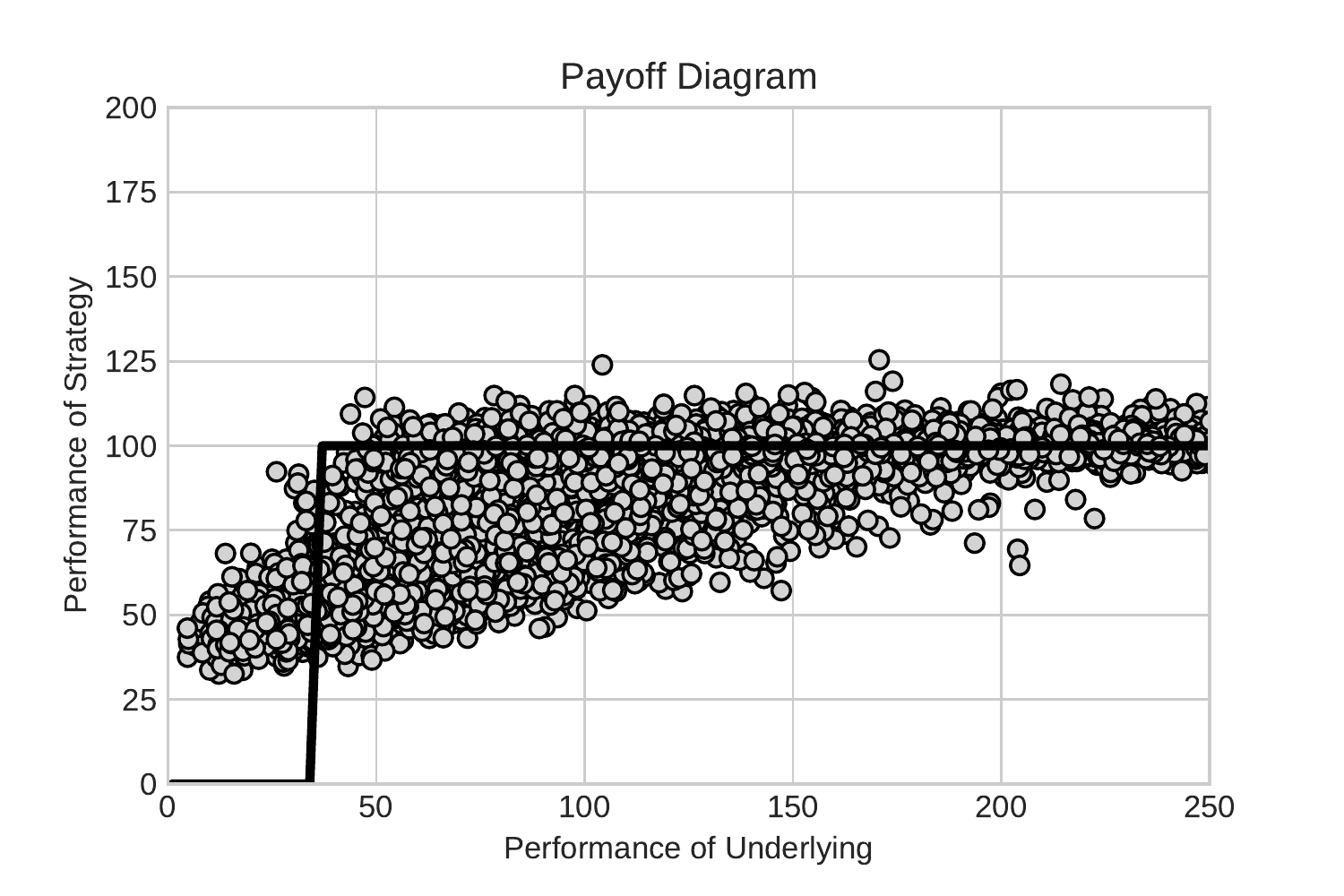}\includegraphics[width=0.5\textwidth]{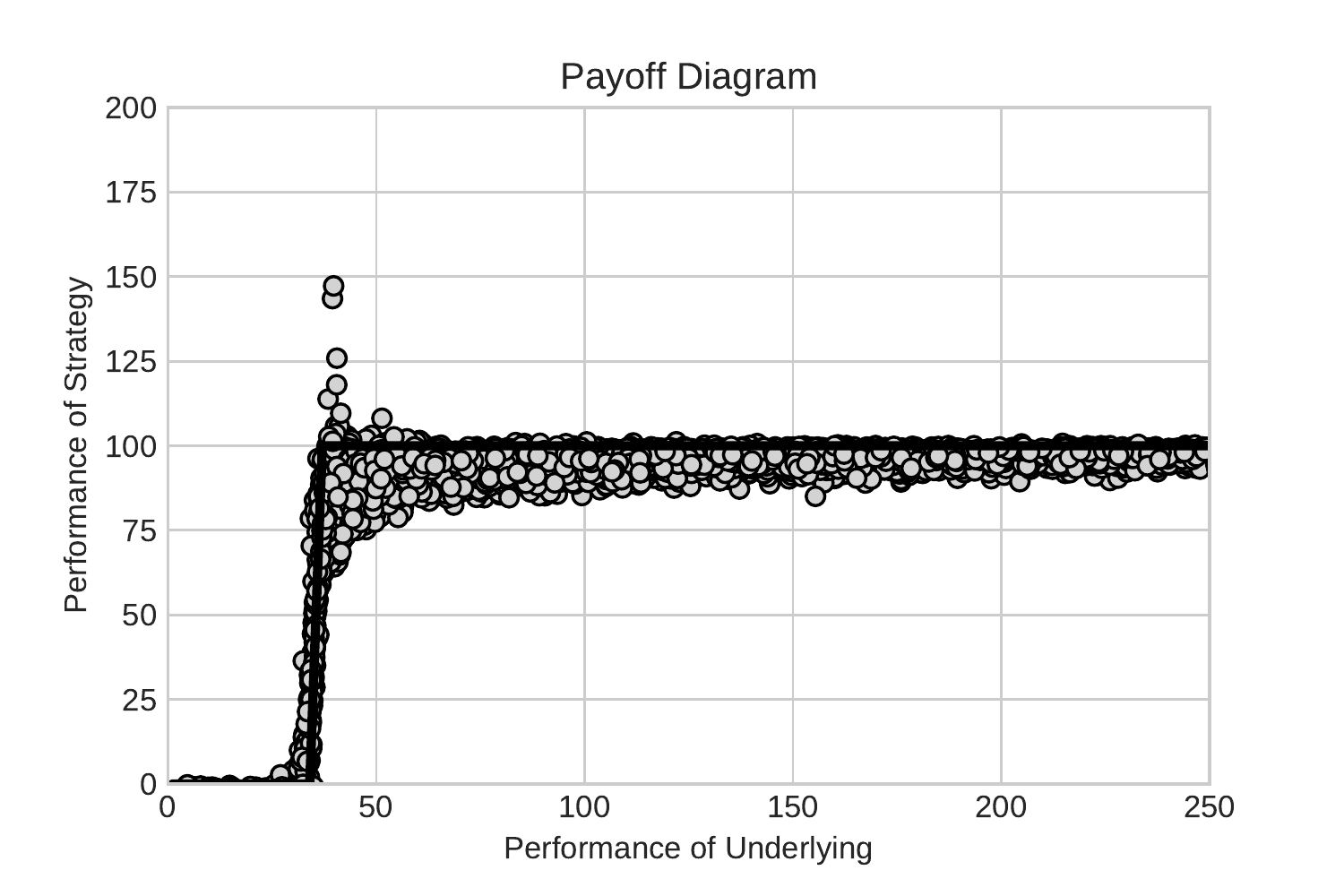}
$p=1.5$\\
\includegraphics[width=0.5\textwidth]{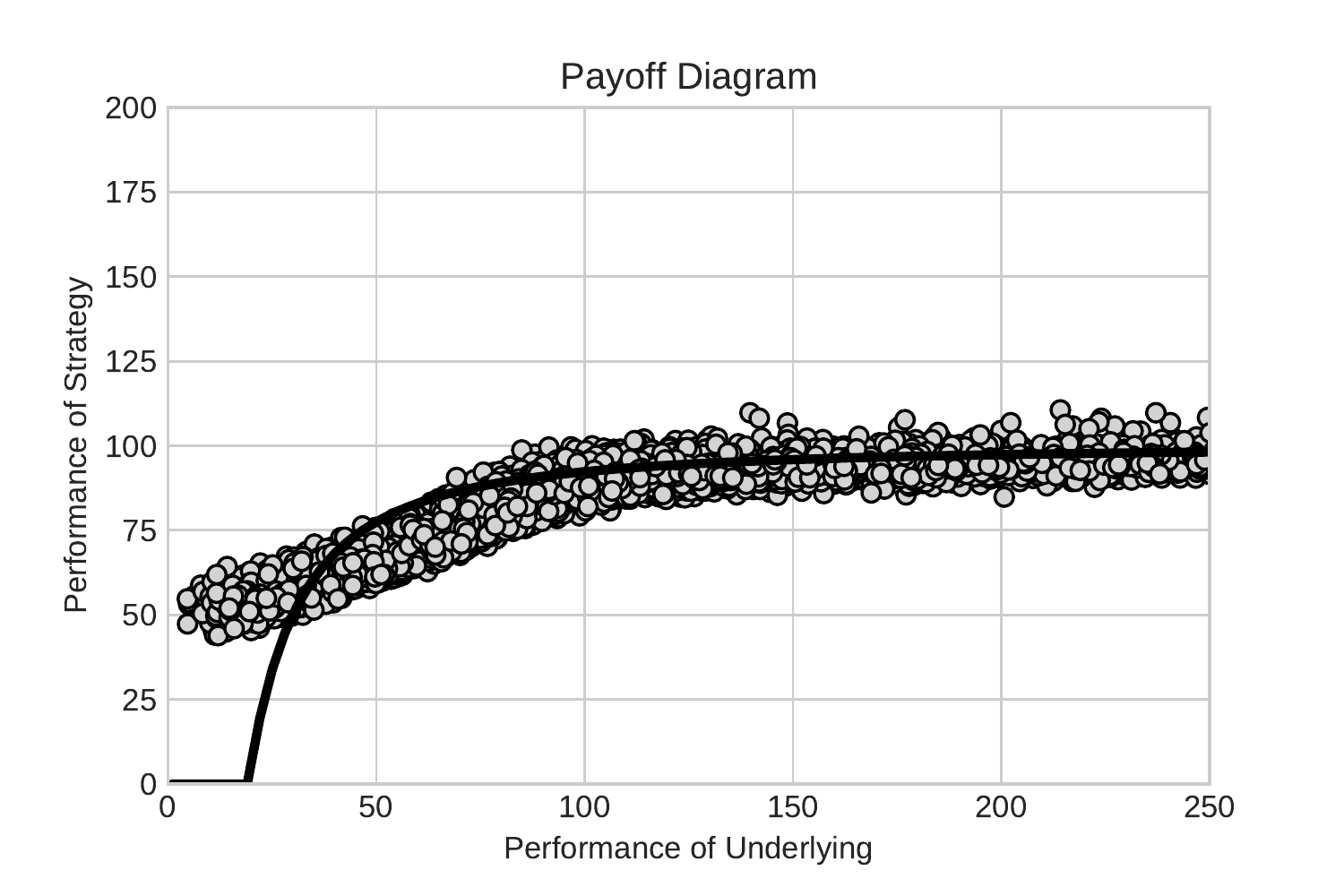}\includegraphics[width=0.5\textwidth]{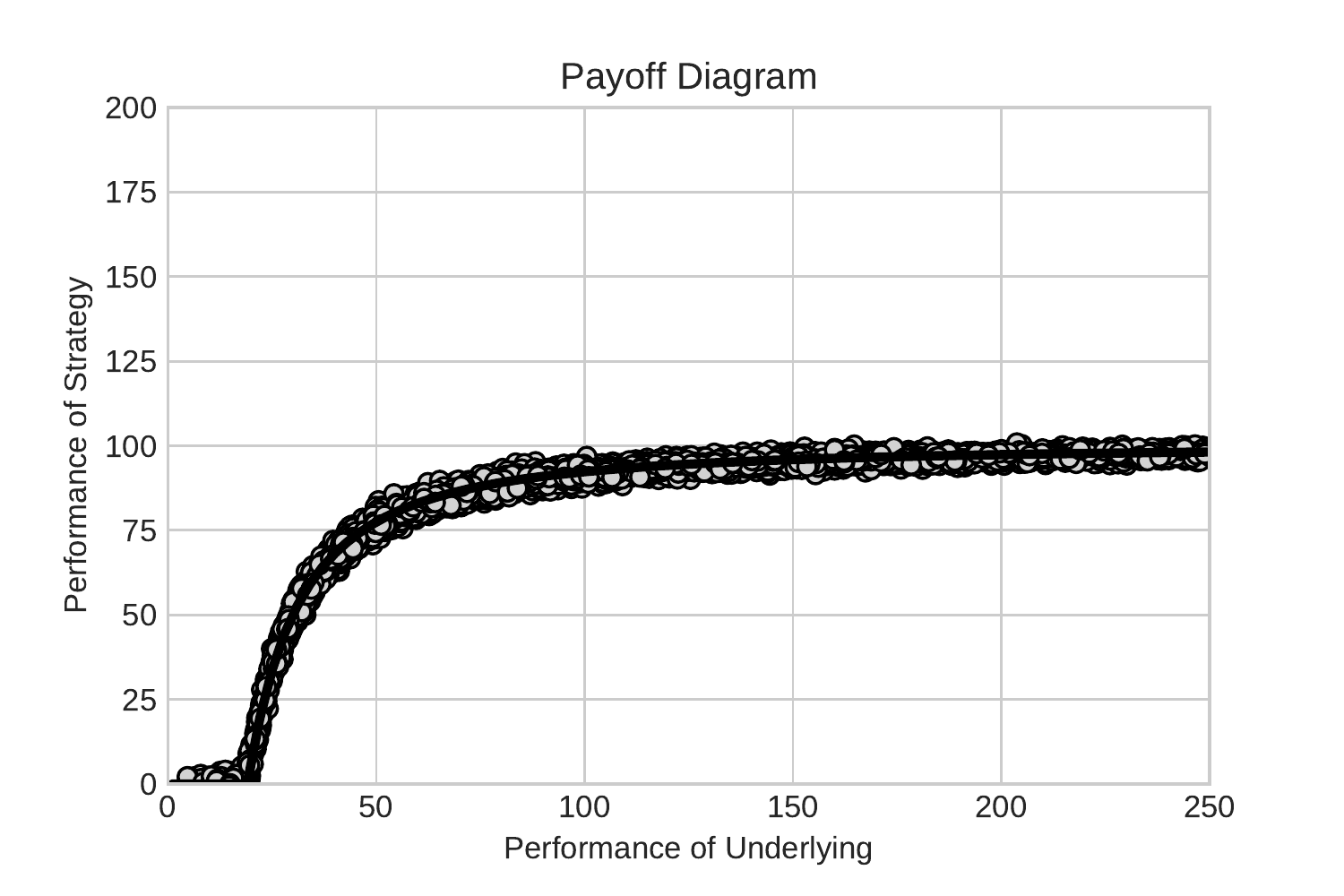}
$p=5$\\
\includegraphics[width=0.5\textwidth]{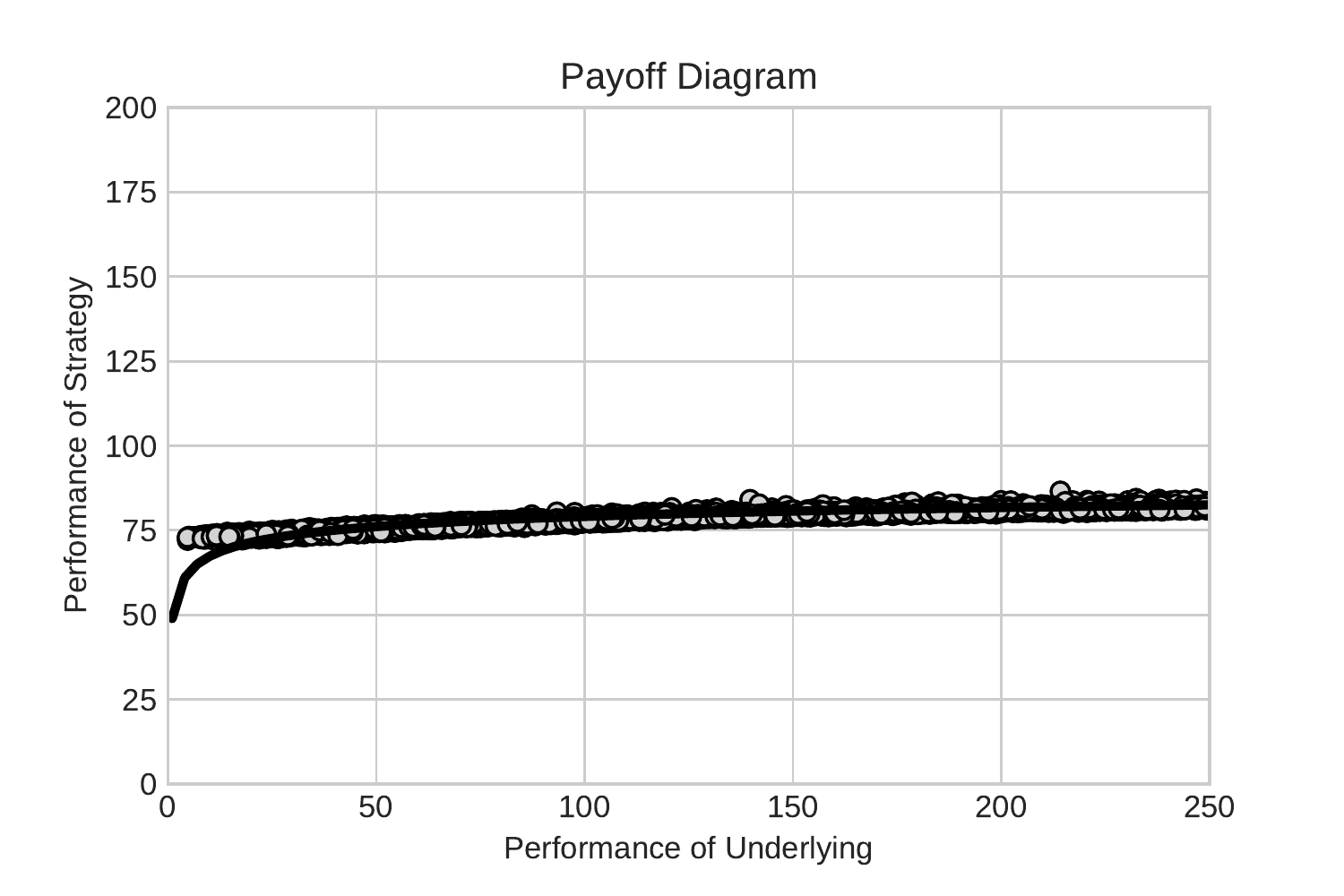}\includegraphics[width=0.5\textwidth]{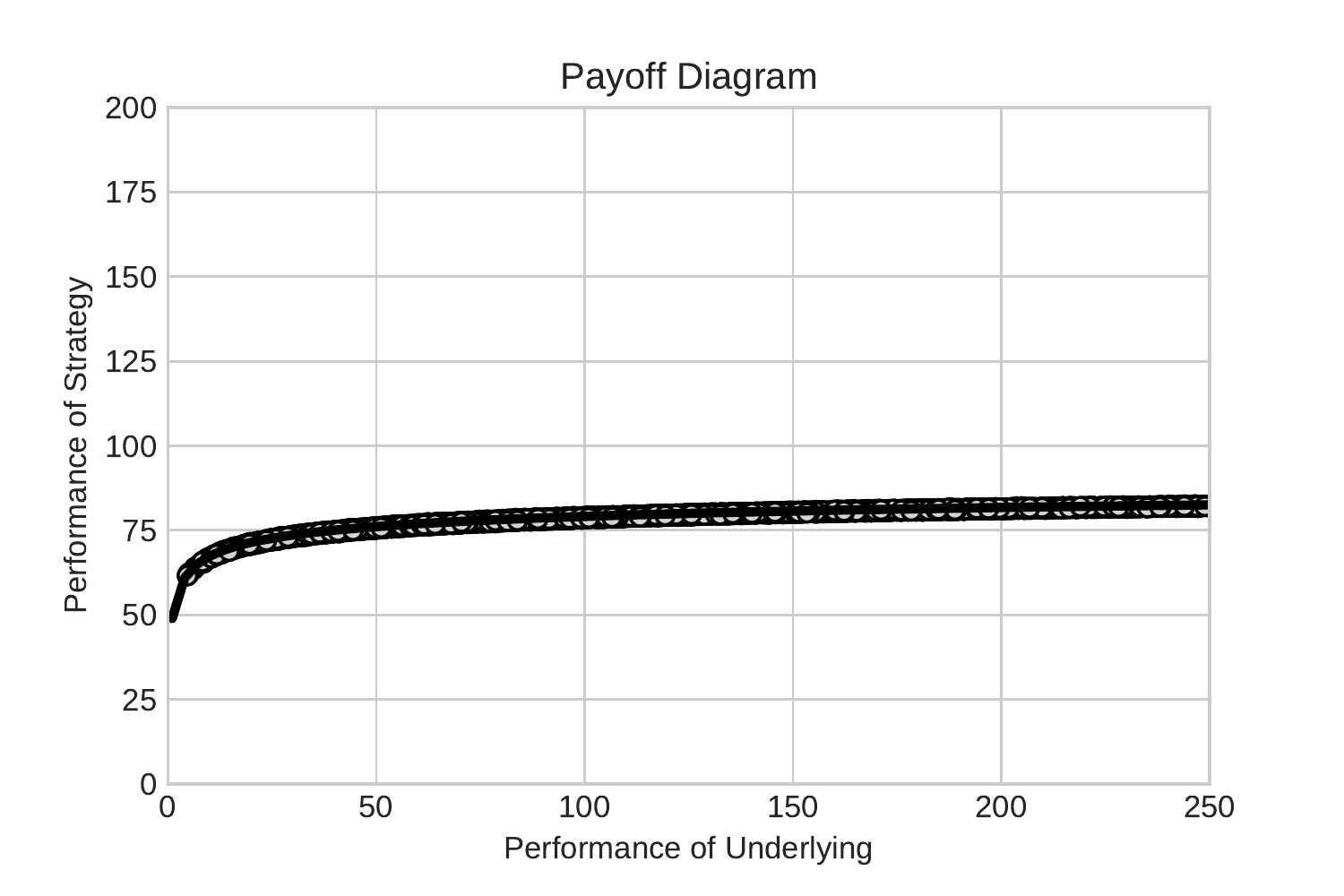}
\caption{For different choices of the risk aversion $p$, the scatter plots depict the final payoffs depending on the performance of the underlying for a trained artificial financial agent in the presence of proportional transaction cost in the left column compared to naively applying the corresponding continuous time optimal delta hedging strategy in the right column. The solid line represents the secondary target payoff originating from the duality result of the continuous-time problem.}
\end{figure}

\begin{table}
\centering
\begin{tabular}{|l||p{1.7cm}|p{1.7cm}|p{1.7cm}|p{1.7cm}|p{1.7cm}|}
\hline
\multirow{2}{*}{Mean}&Theoretical&\multicolumn{2}{c|}{Deep Hedging}&\multicolumn{2}{c|}{Discrete Delta Hedging}\\
\hhline{|~||-|-|-|-|-|}
&\multicolumn{1}{c|}{$\kappa=0$}&\multicolumn{1}{r|}{$\phantom{05}\kappa=0\phantom{0\,0}$}&\multicolumn{1}{c|}{$\kappa=0.005$}&\multicolumn{1}{r|}{$\phantom{05}\kappa=0\phantom{0\,0}$}&\multicolumn{1}{c|}{$\kappa=0.005$}\\
\hline\hline
$p=1$&\multicolumn{1}{r|}{$93.18\phantom{0\,0}$}&\multicolumn{1}{r|}{$91.07\phantom{0\,0}$}&\multicolumn{1}{r|}{$89.10\phantom{0\,0}$}&\multicolumn{1}{r|}{$93.19\phantom{0\,0}$}&\multicolumn{1}{r|}{$88.39\phantom{0\,0}$}\\
$p=1.5$&\multicolumn{1}{r|}{$88.52\phantom{0\,0}$}&\multicolumn{1}{r|}{$88.53\phantom{0\,0}$}&\multicolumn{1}{r|}{$87.44\phantom{0\,0}$}&\multicolumn{1}{r|}{$91.55\phantom{0\,0}$}&\multicolumn{1}{r|}{$87.97\phantom{0\,0}$}\\
$p=5$&\multicolumn{1}{r|}{$80.17\phantom{0\,0}$}&\multicolumn{1}{r|}{$80.50\phantom{0\,0}$}&\multicolumn{1}{r|}{$79.89\phantom{0\,0}$}&\multicolumn{1}{r|}{$80.28\phantom{0\,0}$}&\multicolumn{1}{r|}{$79.81\phantom{0\,0}$}\\
\hline
\multicolumn{6}{c}{}\\
\hline
\multirow{2}{*}{$5\%$-Quantile}&Theoretical&\multicolumn{2}{c|}{Deep Hedging}&\multicolumn{2}{c|}{Discrete Delta Hedging}\\
\hhline{|~||-|-|-|-|-|}
&\multicolumn{1}{c|}{$\kappa=0$}&\multicolumn{1}{r|}{$\phantom{05}\kappa=0\phantom{0\,0}$}&\multicolumn{1}{c|}{$\kappa=0.005$}&\multicolumn{1}{r|}{$\phantom{05}\kappa=0\phantom{0\,0}$}&\multicolumn{1}{c|}{$\kappa=0.005$}\\
\hline\hline
$p=1$&\multicolumn{1}{r|}{$0\phantom{.00}\phantom{0\,0}$}&\multicolumn{1}{r|}{$48.71\phantom{0\,0}$}&\multicolumn{1}{r|}{$48.98\phantom{0\,0}$}&\multicolumn{1}{r|}{$4.05\phantom{0\,0}$}&\multicolumn{1}{r|}{$-2.98\phantom{0\,0}$}\\
$p=1.5$&\multicolumn{1}{r|}{$49.63\phantom{0\,0}$}&\multicolumn{1}{r|}{$54.63\phantom{0\,0}$}&\multicolumn{1}{r|}{$57.17\phantom{0\,0}$}&\multicolumn{1}{r|}{$54.96\phantom{0\,0}$}&\multicolumn{1}{r|}{$48.84\phantom{0\,0}$}\\
$p=5$&\multicolumn{1}{r|}{$73.59\phantom{0\,0}$}&\multicolumn{1}{r|}{$71.96\phantom{0\,0}$}&\multicolumn{1}{r|}{$74.00\phantom{0\,0}$}&\multicolumn{1}{r|}{$73.64\phantom{0\,0}$}&\multicolumn{1}{r|}{$73.11\phantom{0\,0}$}\\

\hline
\multicolumn{6}{c}{}\\
\hline
\multirow{2}{*}{Success Rate}&Theoretical&\multicolumn{2}{c|}{Deep Hedging}&\multicolumn{2}{c|}{Discrete Delta Hedging}\\
\hhline{|~||-|-|-|-|-|}
&\multicolumn{1}{c|}{$\kappa=0$}&\multicolumn{1}{r|}{$\phantom{05}\kappa=0\phantom{0\,0}$}&\multicolumn{1}{c|}{$\kappa=0.005$}&\multicolumn{1}{r|}{$\phantom{05}\kappa=0\phantom{0\,0}$}&\multicolumn{1}{c|}{$\kappa=0.005$}\\
\hline\hline
$p=1$&\multicolumn{1}{r|}{$0.93\phantom{0\,0}$}&\multicolumn{1}{r|}{$0.41\phantom{0\,0}$}&\multicolumn{1}{r|}{$0.36\phantom{0\,0}$}&\multicolumn{1}{r|}{$0.47\phantom{0\,0}$}&\multicolumn{1}{r|}{$0.01\phantom{0\,0}$}\\
$p=1.5$&\multicolumn{1}{r|}{$0\phantom{.00}\phantom{0\,0}$}&\multicolumn{1}{r|}{$0.21\phantom{0\,0}$}&\multicolumn{1}{r|}{$0.12\phantom{0\,0}$}&\multicolumn{1}{r|}{$0.29\phantom{0\,0}$}&\multicolumn{1}{r|}{$0.02\phantom{0\,0}$}\\
$p=5$&\multicolumn{1}{r|}{$0\phantom{.00}\phantom{0\,0}$}&\multicolumn{1}{r|}{$0.00\phantom{0\,0}$}&\multicolumn{1}{r|}{$0.00\phantom{0\,0}$}&\multicolumn{1}{r|}{$0.00\phantom{0\,0}$}&\multicolumn{1}{r|}{$0.00\phantom{0\,0}$}\\
\hline
\multicolumn{6}{c}{}\\
\hline
\multirow{2}{*}{Success Ratio}&Theoretical&\multicolumn{2}{c|}{Deep Hedging}&\multicolumn{2}{c|}{Discrete Delta Hedging}\\
\hhline{|~||-|-|-|-|-|}
&\multicolumn{1}{c|}{$\kappa=0$}&\multicolumn{1}{r|}{$\phantom{05}\kappa=0\phantom{0\,0}$}&\multicolumn{1}{c|}{$\kappa=0.005$}&\multicolumn{1}{r|}{$\phantom{05}\kappa=0\phantom{0\,0}$}&\multicolumn{1}{c|}{$\kappa=0.005$}\\
\hline\hline
$p=1$&\multicolumn{1}{r|}{$0.93\phantom{0\,0}$}&\multicolumn{1}{r|}{$0.89\phantom{0\,0}$}&\multicolumn{1}{r|}{$0.88\phantom{0\,0}$}&\multicolumn{1}{r|}{$0.92\phantom{0\,0}$}&\multicolumn{1}{r|}{$0.88\phantom{0\,0}$}\\
$p=1.5$&\multicolumn{1}{r|}{$0.89\phantom{0\,0}$}&\multicolumn{1}{r|}{$0.87\phantom{0\,0}$}&\multicolumn{1}{r|}{$0.87\phantom{0\,0}$}&\multicolumn{1}{r|}{$0.91\phantom{0\,0}$}&\multicolumn{1}{r|}{$0.88\phantom{0\,0}$}\\
$p=5$&\multicolumn{1}{r|}{$0.80\phantom{0\,0}$}&\multicolumn{1}{r|}{$0.80\phantom{0\,0}$}&\multicolumn{1}{r|}{$0.80\phantom{0\,0}$}&\multicolumn{1}{r|}{$0.80\phantom{0\,0}$}&\multicolumn{1}{r|}{$0.80\phantom{0\,0}$}\\
\hline
\end{tabular}
\vspace{1em}
\caption{Selected empirically derived characteristics of the terminal wealth distribution for $p\in\{1,1.5,5\}$ and $\kappa\in\{0,0.005\}$. The success rate $\P\left[V_T\geq H\right]$ is the counter probability of the shortfall risk. The success ratio is the generalized success rate
$\mathbb{E}\left[\i_{\{V_T\geq H\}}+\frac{V_T}{H}\i_{\{V_T<H\}}\right]$
as defined in \cite[][Definition (2.32)]{follmer_efficient_2000}. 
\\ 
Not only does Deep Hedging yield a flatter right tail in the presence of transaction costs - as can be deduced from the figures for the $5\%$-Quantile - Deep Hedging moreover provides a superior success rate, and can keep up with the success ratio of discrete delta hedging. 
}
\label{tbl:lineup}
\end{table}

\section{Conclusion \& Outlook} 
We have discussed two approaches to goal-based investing in this article. 
The first - analytical - approach yields several explicit continuous dynamic trading strategies that risk-taking, risk-neutral, and risk-averse investors need to implement to maximize their goal-based utilities. 

In the real world, however, continuous-time trading is not feasible. 
We show that this drawback can be addressed with a more flexible deep-hedging approach. 
Not only is this approach well-suited for \textit{discrete} re-balancing, it also allows for the inclusion of transaction costs. 
Curiously, goal-based investing provides a use case for Deep Hedging with a probability-maximizing objective function, due to the problem's equivalence with efficient hedging. 
\\

There are many ramifications of our work on hedging goals that we will investigate elsewhere. 
In particular, open research questions that we will address include:

\begin{itemize}
    \item hedging goals under general market dynamics, e.g., GARCH~\cite{ghalanos_rmgarch_nodate}, or scenarios generated with Generative Adversarial Networks (GAN, cf.~\cite{ni_conditional_2020});
    \item hedging goals with downward protection in the spirit of Section~\ref{sec:down_prot}; 
    \item hedging goals with exogenous income \cite[][Section 7]{browne_reaching_1999} \textit{and} liabilities \cite{browne_survival_1997}; 
    \item beating stochastic benchmarks using deep learning \cite[][Subsection 8.1]{browne_reaching_1999}.
\end{itemize}

\appendix
\section{Proofs} \label{sec:proo}

\subsection{Proofs of the Results with Risk Neutrality and Risk Taking ($p\in[0, 1]$)}

\begin{proof}[Proof of Proposition~\ref{prop:effhed_ndim}]
\cite[][Theorem 9]{leukert_absicherungsstrategien_1999} states the test function 
\begin{align}\label{eq:lem5}
    \varphi_p &= \mathbbm{1}_{\left\{ \frac{\d \mathbb{P}\phantom{^*}}{\d \mathbb{P}^*}  \geq a_p \,H^{1-p} \right \}}, 
\end{align} 
needs to be used to modify the claim $H$. 
Here, $a_p$ is determined implicitly by the capital requirement 
$z = \mathbb{E}^*[\varphi_p\, H]$. To avoid trivial cases, let us assume that $z$ lies within the open interval $\left(0, H_{0, T} \right)$. 
It is straightforward to show that the constant $a_p$ is given by 
\begin{align*}
    a_p &= H^{p-1} \exp\left\{ \sigma_*\, \sqrt{T}\, \Phi^{-1}\left( 1-\frac{z}{H_{0, T}} \right) - \frac{1}{2}{\sigma_*}^2 \ T \right\}. 
\end{align*} 
Let us introduce the density process $(Z_t)_{t\geq 0}$ as
\begin{align}\label{eq:density_process}
    Z_t &= \exp\left\{ - {\pmb{\vartheta}}^\top \, \pmb{W}_t^* + \frac{1}{2} {\sigma_*}^2 \, t \right\}.
\end{align}
Note that $Z_T = \rho_*$, cf. \eqref{eq:radon}.
The density process and the optimal growth portfolio are related via
\begin{align*}
    \log Z_t &= - \left( \log \frac{\Pi_t}{\Pi_0} - \left(r - \frac{1}{2} {\sigma_*}^2\right) t \right) + \frac{1}{2} {\sigma_*}^2 \ t. 
\end{align*}
With these notations, we can show that the value process corresponds to a digital European call option, namely, 
\begin{align*}
    \mathbb{E}^*[\varphi_p H \, | \, \mathcal{F}_t]
    &= \mathbb{P}^*\left[ Z_T \leq \frac{H^{p-1}}{a_p} \, \bigg \rvert \, \mathcal{F}_t \right]
    \\
    &= \mathbb{P}^*\left[ \frac{Z_T}{Z_t} \leq \frac{H^{p-1}}{a_p\ Z_t} \, \bigg \rvert \, \mathcal{F}_t \right]
    \\
    &= \mathbb{P}^*\left[{\pmb{\vartheta}}^\top \left(\pmb{W}_T^*-\pmb{W}_t^*\right) \geq - \left(\log\frac{H^{p-1}}{a_p \, Z_t}-\frac{1}{2}{\sigma_*}^2\, (T-t)\right)\right]
    \\
    &= 1-\Phi\left( \frac{\log a_p - (p-1) \log H + \frac{1}{2}{\sigma_*}^2 \,(T-t) - \left( \log\frac{\Pi_t}{\Pi_0} -\left(r-\frac{1}{2} {\sigma_*}^2 \right) t \right) + \frac{1}{2}{\sigma_*}^2 \,t}{\sigma_* \,\sqrt{T-t}} \right)
    \\
    &= 1-\Phi\left( \frac{ \Phi^{-1}\left(1-\frac{ z}{H_{0, T}}\right)\sigma_* \,\sqrt{T} - \left( \log\frac{\Pi_t}{\Pi_0} -\left(r-\frac{1}{2} {\sigma_*}^2\right) t \right)}{\sigma_* \,\sqrt{T-t}} \right)
    \\
    &= \Phi\left( \frac{\log\frac{\Pi_t}{\Pi_0}  - \left( r-\frac{1}{2} {\sigma_*}^2\right) t - \Phi^{-1}\left(1-\frac{z}{H_{0, T}}\right)\sigma_*\,\sqrt{T}}{\sigma_* \, \sqrt{T-t}} \right)
    \\
    &= \Phi\left( \frac{\log\frac{\Pi_t}{K^*} + \left(r-\frac{1}{2}{\sigma_*}^2\right)(T-t)}{\sigma_*\, \sqrt{T-t}} \right), 
\end{align*}
where the strike $K^*$ is given by 
\begin{align*}
    \log K^* &= \log \Pi_0 + \left(r-\frac{1}{2}{\sigma_*}^2\right)\,T-\sigma_*\,\sqrt{T}\,\Phi^{-1}\left( \frac{z}{H_{0, T}} \right).
\end{align*}
This furthermore shows that the solution for the multivariate problem of minimizing the expected shortfall is identical to the one derived by Browne in the case of maximizing the probability of reaching an investment goal \cite{browne_reaching_1999}.
\end{proof}

\begin{proof}[Proof of Corollary~\ref{cor:effhed_1dim}]
By virtue of \eqref{eq:radon}, we can express the test function 
$\varphi_p$ as the indicator function
\begin{align*}
    \varphi_p
    &= \mathbbm{1}{\left\{ \rho_*  \leq \frac{H^{p-1}}{a_p} \right \}} 
    = \mathbbm{1}{\left\{ \exp\left\{ \frac{1}{2}\vartheta^2\, T-\vartheta\,W_T^* \right\} \leq \frac{H^{p-1}}{a_p} \right \}} 
    \\
    &= \mathbbm{1}{\left\{ W_T^* \geq \frac{1}{\vartheta}\left( \frac{1}{2}\vartheta^2\, T + \log a_p + (1-p) \log H \right)  \right \}}. 
\end{align*}
Hence, for a standard normal random variate $Y$, 
\begin{align*}
   R_{T, 0}\, z &= R_{T, 0}\, \mathbb{E}^*[\varphi_p H]
    = H \, \mathbb{P}^*\left[ \sqrt{T}\, Y \geq \frac{1}{\vartheta}\left( \frac{1}{2}\vartheta^2\, T + \log a_p + (1-p) \log H\right) \right]
    \\
    &= H \left( 1 - \Phi\left(\frac{\frac{1}{2}\vartheta^2\, T+\log a_p + (1-p) \log H  }{\vartheta \, \sqrt{T}}\right)\right). 
\end{align*}
Thus, 
\begin{align*}
    a_p &= H^{p-1} \exp\left\{ \vartheta \,\sqrt{T}\Phi^{-1}\left(1-\frac{z}{H_{0, T}}\right) - \frac{1}{2}\vartheta^2\, T \right\}. 
\end{align*}
It can be shown that $\rho_* = k \ {X_T}^{-\alpha}$, for a real constant $k$. 
In fact, 
\begin{align*}
    {X_T}^{-\alpha} &= {x_0}^{-\alpha}\ \exp\left\{ - \alpha \left(\mu - \frac{\sigma^2}{2}\right) T -\alpha\sigma W_T \right\} 
    \\
    &= {x_0}^{-\alpha} \exp\left\{ - \alpha \left( \frac{\mu - r}{2} + \frac{ \mu + r - \sigma^2}{2}\right) T -\vartheta\, W_T \right\}
    \\ 
    &= {x_0}^{-\alpha} \exp\left\{ - \frac{\alpha(\mu + r - \sigma^2)T}{2} \right\} \underbrace{\exp\left\{- \frac{1}{2}\vartheta^2\, T  -\vartheta W_T \right\}}_{=\rho^*}, 
\intertext{and hence}
    k &= {x_0}^\alpha \exp\left\{ \frac{\alpha(\mu + r - \sigma^2)T}{2} \right\}.
\end{align*}
The test function $\varphi_p$ in \eqref{eq:lem5} can therefore be rewritten as 
\begin{align*}
    \varphi_p 
    &= \mathbbm{1} \left\{ k\, {X_T}^{-\alpha} \leq \frac{H^{p-1}}{a_p} \right\}
    = \mathbbm{1}\left\{ X_T \geq \sqrt[\alpha]{k\, a_p\, H^{1-p}} \right\}
    \\
    &= \mathbbm{1}\left\{ X_T \geq x_0 \exp\left\{ \frac{\mu+r-\sigma^2}{2}\,T + \Phi^{-1}\left(1-\frac{ z}{H_{0, T}}\right) \sigma\,\sqrt{T} - \frac{\mu - r}{2}\,T \right\} \right\}
    \\
    &= \mathbbm{1}\left\{ X_T \geq x_0 \exp\left\{ \left( r - \frac{\sigma^2}{2}\right) T + \Phi^{-1}\left(1-\frac{ z}{H_{0, T}}\right) \sigma\,\sqrt{T} \right\} \right\}
    \\
    &= \mathbbm{1}\left\{ X_T \geq x_0 \exp\left\{ \left( r - \frac{\sigma^2}{2}\right) T - \Phi^{-1}\left(\frac{ z}{H_{0, T}}\right) \sigma\,\sqrt{T} \right\} \right\}.
\end{align*}
\\
Let $Z_T := \rho_*$, with the density process $Z=(Z_t)_{t\in [0, T]}$ defined as
\begin{align*}
    \log Z_t &= -\frac{\vartheta}{\sigma} \left(\log\frac{X_t}{x_0} -\left(r- \frac{1}{2}\sigma^2\right) t\right) + \frac{1}{2}\vartheta^2  t. 
\end{align*}
Then 
we have that (cf.~\cite[][Corollary 2.8]{xu_minimizing_2004})
\begin{align*}
    \mathbb{E}^*[\varphi_p H \, | \, \mathcal{F}_t]
    &= \mathbb{P}^*\left[ Z_T \leq \frac{H^{p-1}}{a_p} \, \bigg \rvert \, \mathcal{F}_t \right]
    \\
    &= \mathbb{P}^*\left[ \frac{Z_T}{Z_t}\, Z_t \leq \frac{H^{p-1}}{a_p} \, \bigg \rvert \, \mathcal{F}_t \right]
    \\
    &= \mathbb{P}^*\left[ \frac{Z_T}{Z_t} \leq \frac{H^{p-1}}{a_p\ Z_t} \, \bigg \rvert \, \mathcal{F}_t \right]
    \\
    &= \mathbb{P}^*\left[{W_T}^*-{W_t}^*\geq-\frac{1}{\vartheta}\left(\log\left(\frac{H^{p-1}}{a_pZ_t}\right)-\frac{1}{2}\vartheta^2(T-t)\right)\right]
    \\
    &= 1-\Phi\left( \frac{\log a_p + (1-p)\log H + \frac{1}{2}\vartheta^2 (T-t) - \frac{\vartheta}{\sigma}\left( \log\frac{X_t}{x_0} -\left(r- \frac{\sigma^2}{2}\right) t \right) + \frac{1}{2}\vartheta^2 t}{\vartheta \sqrt{T-t}} \right)
    \\
    &= 1-\Phi\left( \frac{ \Phi^{-1}\left(1-\frac{ z}{H_{0, T}}\right)\vartheta \sqrt{T} - \frac{\vartheta}{\sigma}\left( \log\frac{X_t}{x_0} -\left(r- \frac{\sigma^2}{2}\right) t \right)}{\vartheta \,\sqrt{T-t}} \right)
    \\
    &= \Phi\left( \frac{\log\frac{X_t}{x_0}  - \left( r- \frac{\sigma^2}{2}\right) t - \Phi^{-1}\left(1-\frac{ z}{R_{0, T}\,H}\right)\sigma\sqrt{T}}{\sigma \sqrt{T-t}} \right)
    \\
    &= \Phi\left( \frac{\log\frac{X_t}{K^*} + \left(r-\frac{\sigma^2}{2}\right)(T-t)}{\sigma \sqrt{T-t}} \right), 
\end{align*}
where $K^*$ is given by 
\begin{align*}
    \log K^* = \log x_0 + \Phi^{-1}\left(1-\frac{z}{H_{0, T}}\right) \sigma \,\sqrt{T} + \left(r-\frac{\sigma^2}{2}\right) T. 
\end{align*}
The modified claim $\varphi_p H$ thus corresponds to a digital call option with strike $K^*$; cf.~\eqref{eq:bs_price_digital}.
\end{proof}
\begin{remark}
Note that, as $z$ approaches $0$, the inverse cumulative distribution function diverges to $+\infty$, so that $K^*$ tends to $\infty$ and, as a consequence, the (initial) value of the modified claim $V_t$ vanishes.

Conversely, as $z$ approaches $H_{0, T}$ from below, $K^*$ diverges to $-\infty$, so that $\varphi_p \rightarrow \mathbbm{1}_{\R_+}$: in the limit, the modified claim coincides with the original one. 
\end{remark}

\subsection{Proofs of the Results with Risk Aversion}

\begin{proof}[Proof of Proposition~\ref{prop:lowmom_ndim}]
Recall that, in the case of increasing risk aversion, we need to consider the problem \eqref{eq:modifH}. 
For this purpose, we note that the density process $(Z_t)_{t\in[0, T]}$ (cf. \eqref{eq:density_process}) relates to the optimal-growth portfolio via
\begin{align*}
    Z_T &= \rho_* = \frac{\Pi_0}{R_{0, T}\,\Pi_T}. 
\end{align*}
The modified claim of \eqref{eq:modifH} thus takes the form 
\begin{align*}
    \varphi_p &= \left(1-a_p\left(\frac{\Pi_0}{R_{0, T}\,\Pi_T}\right)^{p'}\right)_+, 
\end{align*}
where we have used the shorthand $p'=1/(p-1)$. 
This equation in turn can be rewritten as 
\begin{align*}
    \varphi_p &= \left(1-\left(\frac{L}{\Pi_T}\right)^{p'}\right)\, \mathbbm{1}_{\{ \Pi_T \geq L \}}, 
\end{align*}
where the threshold is given by $L:=\sqrt[p']{a_p}\, R_{T, 0}\, \Pi_0$.
This claim consists of a European digital option that is modified by a factor. 
The difference now, however, is that the digital option is a contingent claim on the optimal growth portfolio, whose wealth at time $t$ is given by $\Pi_t$.
\\
Calculations analogous to those in the case of a single risky asset (cf.~the proof of Corollary~\ref{cor:lowmom_1dim} below) show that the modified claim on the optimal-growth portfolio takes the form specified in Proposition \ref{prop:lowmom_ndim}.
\end{proof}
\begin{proof}[Proof of Corollary~\ref{cor:lowmom_1dim}]
The modified claim \eqref{eq:modifH} in this case reads as 
$$
\varphi_p = \left(
1-a_p\, {\rho_*}^{p'}\right)_+.
$$
Recall from the proof of Corollary~\ref{cor:effhed_1dim} that
$$
\rho_* = \frac{k}{{X_T}^{\alpha}}, 
$$
where 
\[
    k = {x_0}^\alpha\exp\left\{ \frac{\alpha (\mu + r - \sigma^2)T}{2}\right\}. 
\]
Therefore, 
\begin{align}\label{eq:expo}
    \varphi_p = \left( 1 - a_p\frac{k^{p'}}{ {X_T}^{\alpha p'}}\right)\, \mathbbm{1}_{\left\{X_T \geq {a_p}^\frac{p-1}{\alpha}k^{\frac{1}{\alpha}}\right\}}. 
\end{align}
Let us denote the threshold by $L := {a_p}^\frac{p-1}{\alpha}k^{\frac{1}{\alpha}}$. 
Thus Equation~\eqref{eq:expo} can be rewritten as 
\begin{align}\label{eq:fraclx}
    \varphi_p = \left( 1-\left(\frac{L}{X_T}\right)^{\alpha_p} \right)\, \mathbbm{1}_{\{ X_T \geq L \}}. 
\end{align}
Defining the function
    $f_p(y) := \left( 1-\frac{L^{\alpha_p}}{y^{\alpha_p}} \right)\, \mathbbm{1}_{\{ y \geq L \}}$
for $y \in \R$, we set 
\begin{align*}
    V_t 
    &= \mathbb{E}^*[\varphi_p H\,|\, \mathcal{F}_t] 
    = H_{T,t} \, \mathbb{E}^*[f_p(X_t\, \exp[\sigma\, (W_T^*-W_t^*) + (r-\sigma^2/2)\,(T-t)\,|\, \mathcal{F}_t]
    \\ 
    &=: H_{T, t}\, F_p(t, X_t), 
\end{align*}
so that, for $\tau:=T-t$, 
\begin{align}\label{eq:digiCall-rem}
    H_{T, t}\,F_p(t, x) 
    &= \int_\R f_p(x\exp[\sigma \sqrt{\tau}\, y + (r-\sigma^2/2)\, \tau])\, \exp(-y^2/2)\frac{\d y}{\sqrt{2\pi}}
    \nonumber
    \\
    &= \Phi(d_-(t; x, L)) - \frac{L^{\alpha_p}}{x^{\alpha_p}} \int_{-d_-(t; x, L)}^\infty \exp \left\{ -\alpha_p(\sigma\sqrt{\tau} y + (r-\sigma^2/2)\tau) \right\}\, e^\frac{-y^2}{2}\frac{\d y}{\sqrt{2\pi}}
    \nonumber
    \\ 
    &= \Phi(d_-(t; x, L)) - \frac{L^{\alpha_p}}{x^{\alpha_p}}
    \exp\left\{  \alpha_p \left(\alpha_p+1\right) (\sigma^2/2-r) \tau \right\} \Phi\left( d_-(t; x, L)-\alpha_p \,\sigma \,\sqrt{\tau} \right). 
\end{align}
The threshold $L$ is determined implicitly by the initial endowment $z$ via
\begin{align*}
    z &= \mathbb{E}^*[\varphi_p H] = H_{0, T}\,F_p(0, x_0) 
    \\
    &= H_{0,T}\, \left(\Phi(d_-(0; x_0, L)) -  \frac{L^{\alpha_p}}{x_0^{\alpha_p}}\, \exp\left\{ \alpha_p\left(\alpha_p+1\right) (\sigma^2/2-r) \,T \right\}\, \Phi\left( d_-(0; x, L)-\alpha_p \,\sigma\, \sqrt{T} \right)\right). 
\end{align*}
\end{proof}

\subsection{Proofs of the Results with Downward Protection}
\begin{proof}[Proof of Proposition~\ref{prop:down_prot}]
This follows from applying the results in~\cite[][Section 8]{browne_reaching_1999}. 
\\
\end{proof}
\printbibliography[title={References}]

\end{document}